\documentclass{llncs}

\usepackage[utf8]{inputenc}
\usepackage[T1]{fontenc}
\usepackage{amsfonts}
\usepackage{amsmath}
\usepackage{amssymb}
\usepackage{float}
\usepackage{url}
\usepackage{hyperref}
\usepackage{subfigure}
\usepackage[margin=2cm]{geometry}

\usepackage{tikz}
\usetikzlibrary{arrows, automata, decorations.pathmorphing, fit, shapes}
\tikzstyle{every picture}=[
  >=stealth',
  shorten >=1pt,
  shorten <=1pt,
  node distance=1.5cm,
  auto,
  bend angle=45,
  initial text=,
  every state/.style={
	rounded rectangle, 
	rounded rectangle arc length=90,
  	inner sep=0.5mm,
  	minimum size=5mm
  },
]

\allowdisplaybreaks[1]

\begin{document}

\title{On the Hierarchy of Block Deterministic Languages}
  
\author{Pascal Caron \and Ludovic Mignot \and Clément Miklarz} 

\institute{
  LITIS, Université de Rouen, 76801 Saint-Étienne du Rouvray Cedex, France\\
  \email{\{pascal.caron,ludovic.mignot,clement.miklarz1\}@univ-rouen.fr}
}
  
\maketitle

  \begin{abstract}
    A regular language is $k$-lookahead deterministic (resp. $k$-block deterministic) if it is specified by a $k$-lookahead deterministic (resp. $k$-block deterministic) regular expression.
	These two subclasses of regular languages have been respectively introduced by Han and Wood ($k$-lookahead determinism) and by Giammarresi \emph{et al.} ($k$-block determinism) as a possible extension of one-unambiguous languages defined and characterized by Brüggemann-Klein and Wood.

	In this paper, we study the hierarchy and the inclusion links of these families.
	We first show that each $k$-block deterministic language is the alphabetic image of some one-unambiguous language. 
	Moreover, we show that the conversion from a minimal DFA of a $k$-block deterministic regular language to a $k$-block deterministic automaton not only requires state elimination, and that the proof given by Han and Wood of a proper hierarchy in $k$-block deterministic languages based on this result is erroneous.
	Despite these results, we show by giving a parameterized family that there is a proper hierarchy in $k$-block deterministic regular languages. 
	We also prove that there is a proper hierarchy in $k$-lookahead deterministic regular languages by studying particular properties of unary regular expressions.
	Finally, using our valid results, we confirm that the family of $k$-block deterministic regular languages is strictly included into the one of $k$-lookahead deterministic regular languages by showing that any $k$-block deterministic unary language is one-unambiguous.
  \end{abstract} 

\section{Introduction}\label{se:int}

	A Document Type Definition (DTD) containing a grammar is used to know whether an XML file fits some specification. These grammars are made of rules whose right-hand part is a restricted regular expression.
	Brüggemann-Klein and Wood have formalized these regular expressions and have shown that the set of languages specified is strictly included in the set of regular ones.
	The distinctive aspect of such expressions is the one-to-one correspondence between each letter of the input word and a unique position in them.
	The resulting Glushkov automaton is deterministic.
	The languages specified are called one-unambiguous regular languages.

	Several extensions of one-unambiguous expressions have been considered:
\begin{itemize}
	\item $k$-block deterministic regular expressions~\cite{GMW01} are such that while reading an input word, there is a one-to-one correspondence between the next at most $k$ input symbols and the same number of symbols of the expression.
These expressions have particular Glushkov automata.
The transitions of these automata can be labeled by words of length at most $k$ and for every couple of words labeling two output transitions of a single state, these words are not prefix from each other. 
	\item $k$-lookahead deterministic regular expressions form another generalization.
This time, the reading of the next $k$ symbols of the input word allows one to know the next position in the expression. 
This extension has been proposed in~\cite{HW08}. 
  \item $(k,l)$-unambiguous regular expressions~\cite{CFM14} is another extension of one-unam\-bi\-guity, where the next $k$ symbols may induce several paths, but with at most one common state.
\end{itemize}

	These three families of expressions fit together as families of languages in the way that a language is $k$-block deterministic (resp. $k$-lookahead deterministic, $(k,l)$-unambiguous) if there exists a $k$-block deterministic (resp. $k$-lookahead deterministic, $(k,l)$-unambiguous) expression to represent it. 

	In~\cite{HW08}, Han and Wood show that there is a proper hierarchy in block deterministic languages and there is a strict inclusion of the family of $k$-block deterministic languages into the one of $k$-lookahead deterministic languages.
	However, they based their proofs on an erroneous statement due to Giammaresi \emph{et al.}~\cite{GMW01}, invalidating them.
	In this paper, we first show that there is indeed a proper hierarchy in block deterministic languages by giving our own parameterized family.
	Then, we show that there is also a proper hierarchy in $k$-lookahead deterministic languages by studying the structural properties of unary Glushkov automata.
	Finally, using our valid results, we demonstrate that the family of $k$-block deterministic languages is strictly included into the one of $k$-lookahead deterministic languages by showing that any $k$-block deterministic unary language is also one-unambiguous.

	Preliminaries are gathered in Section~\ref{se:pre}.
	In Section~\ref{se:PrevRes}, we recall several results from~\cite{GMW01,HW08} on which we question their truthfulness.
	Indeed, we show in Section~\ref{se:kbd_witness} that, due to an erroneous statement of Lemma~\ref{le:GMWH}, the witness family given as a proof of Theorem~\ref{th:HW} is invalid; and present an alternative family, proving the infinite hierarchy of $k$-block deterministic regular languages w.r.t. $k$.
	In Section~\ref{se:kld_witness}, we give another witness family to prove that there is also an infinite hierarchy in $k$-lookahead deterministic regular languages w.r.t. $k$.
	Then, in Section~\ref{se:inclusion_kbd_kld}, we give our own proof that $k$-block deterministic regular languages are a proper subfamily of $k$-lookahead deterministic regular languages w.r.t. $k$.

\section{Preliminaries}\label{se:pre}

\subsection{Languages and Automata Basics}
	Let $\Sigma$ be a non-empty finite \emph{alphabet}.
A \emph{word $w$ over $\Sigma$} is a finite sequence of symbols from $\Sigma$.
The \emph{length} of a word $w$ is denoted by $|w|$, and the \emph{empty word} is denoted by $\varepsilon$.
Let $p, f, s, w \in \Sigma^*$ be words such that $w = pfs$, then $p$ is a \emph{prefix} of $w$ and $f$ is a \emph{subword} of $w$.
The set of all prefixes (respectively subwords) of $w$ is denoted by $\mathrm{Pref}(w)$ (respectively $\mathrm{Subw}(w)$).

	Let $\Sigma^*$ denote the set of all words over $\Sigma$.
A \emph{language over $\Sigma$} is a subset of $\Sigma^*$.
Let $L$ and $L^{\prime}$ be two languages over $\Sigma$. The following operations are defined:
\begin{itemize}
	\item \emph{the union}: $L \cup L^{\prime} = \{w \mid w \in L \vee w \in L^{\prime}\}$
	\item \emph{the concatenation}: $L \cdot L^{\prime} = \{w \cdot w^{\prime} \mid w \in L \wedge w^{\prime} \in L^{\prime}\}$
	\item \emph{the Kleene star}: $L^* = \bigcup_{k \in \mathbb{N}} L^k$ with $L^0 = \{\varepsilon\}$ and $L^{k+1} = L \cdot L^k$
\end{itemize}

	A \emph{regular expression over $\Sigma$} is built from $\emptyset$ (the empty set), $\varepsilon$, and symbols in $\Sigma$ using the binary operators $+$ and $\cdot$, and the unary operator $^*$.
	The \emph{language} $\mathrm{L}(E)$ \emph{specified by a regular expression $E$} is defined as follows:
\begin{align*}
	\mathrm{L}(\emptyset) &= \emptyset, & \mathrm{L}(\varepsilon) &= \{\varepsilon\}, & \mathrm{L}(a) &= \{a\},\\
	\mathrm{L}(F + G) &= \mathrm{L}(F) \cup \mathrm{L}(G),\  & \mathrm{L}(F \cdot G) &= \mathrm{L}(F) \cdot \mathrm{L}(G),\  & \mathrm{L}(F^*) &= \mathrm{L}(F)^*,
\end{align*}
with $a \in \Sigma$, and $F$, $G$ some regular expressions over $\Sigma$.
	Given a language $L$, if there exists a regular expression $E$ such that $\mathrm{L}(E) = L$, then $L$ is a \emph{regular language}.
	A regular expression is trimmed if it is equal to $\emptyset$ or does not contain any occurrence of $\emptyset$.
	We consider only trimmed regular expressions in the rest of this paper.

	A \emph{finite automaton} $A$ is a 5-tuple $(\Sigma, Q, I, F, \delta)$ where:
$Q$ is a finite set of states, $I \subset Q$ is the set of initial states, $F \subset Q$ is the set of final states, and $\delta \subset Q \times \Sigma \times Q$ is a set of transitions.
The set $\delta$ is equivalent to a function of $Q \times \Sigma \rightarrow 2^Q$ : $(p, a, q) \in \delta \Longleftrightarrow q \in \delta (p, a)$.
This function can be extended to $2^Q \times \Sigma^* \rightarrow 2^Q$ as follows: for any subset $Q^{\prime} \subset Q$, for any symbol $a \in \Sigma$, for any word $w \in \Sigma^*$: $\delta(Q^{\prime}, \varepsilon) = Q^{\prime}$, $\delta(Q^{\prime}, a) = \bigcup_{q \in Q^{\prime}} \delta (q, a)$, $\delta(Q^{\prime}, a \cdot w) = \delta(\delta(Q^{\prime}, a), w)$; finally, we set $\delta(q, w)=\delta(\{q\}, w)$.

	A set $O \subset Q$ is called an \emph{orbit} if it is a strongly connected component.
An orbit is \emph{trivial} if it consists of only one state and there is no transition from it to itself in $A$.
The set of orbits of $A$ is denoted by $\mathcal{O}_A$.
	Let $O \in \mathcal{O}_A$ be an orbit and $p \in O$ be a state. The state $p$ is an \emph{out-gate of $O$} (respectively an \emph{in-gate of $O$}) if $(p \in F) \vee (\exists a \in \Sigma, \exists q \in (Q \setminus O), q \in \delta(p, a))$ (respectively if $(p \in I) \vee (\exists a \in \Sigma, \exists q \in (Q \setminus O), p \in \delta(q, a))$).
	The set of out-gates (respectively in-gates) of $O$ is denoted by $\mathrm{G_{out}}(O)$ (respectively $\mathrm{G_{in}}(O)$).

	The \emph{language} $\mathrm{L}(A)$ \emph{recognized by $A$} is the set $\{w \in \Sigma^* \mid \delta(I, w) \cap F \neq \emptyset\}$.
Two automata are \emph{equivalent} if they recognize the same language.
The \emph{right language of a state $q$ of $A$} is denoted by $\mathrm{L}_{q}(A) = \{w \in \Sigma^* \mid \delta(q, w) \cap F \neq \emptyset\}$.
Two states are \emph{equivalent} if they have the same right language.

	An automaton $A = (\Sigma, Q, I, F, \delta)$ is trimmed if $\forall q \in Q, \exists w_p, w_s \in \Sigma^*, q \in \delta(I, w_p) \wedge \delta(q, w_s) \cap F \neq \emptyset$.
	If an automaton is not trimmed, it is possible to compute an equivalent trimmed automaton by getting rid of any useless state.
	We consider only trimmed automata in the rest of this paper.

	An automaton $A = (\Sigma, Q, I, F, \delta)$ is \emph{standard} if $|I| = 1$ and $\forall q \in Q, \forall a \in \Sigma, \delta(q, a) \cap I = \emptyset$.
If $A$ is not a standard automaton, then it is possible to compute an equivalent standard automaton $(\Sigma, Q_s, I_s, F_s, \delta_s)$ as follows:
\begin{itemize}
	\item $Q_s = Q \cup \{i_s\}$ with $i_s \notin Q$
	\item $I_s = \{i_s\}$
	\item $F_s = F \cup \{i_s\}$ if $I \cap F \neq \emptyset$, $F$ otherwise
	\item $\delta_s = \delta \cup \{(i_s, a, q) \mid \exists i \in I, (i, a, q) \in \delta\}$
\end{itemize}
This operation is called \emph{standardization}.

	An automaton $A = (\Sigma, Q, I, F, \delta)$ is \emph{deterministic} if $|I| = 1$ and $\forall t_1 = (p, a, q_1)$, $t_2 = (p, b, q_2) \in \delta, (t_1 \neq t_2) \Longrightarrow (a \neq b)$.
	If $A$ is not deterministic, then it is possible to compute an equivalent deterministic automaton by using the powerset construction described in~\cite{RS59}.
	
	A deterministic automaton $A = (\Sigma, Q_A, \{i_A\}, F_A, \delta_A)$ is \emph{minimal} if there is no equivalent deterministic automaton $B = (\Sigma, Q_B, \{i_B\}, F_B, \delta_B)$ such that $|Q_B| < |Q_A|$.
	If $A$ is not minimal, then it is possible to compute an equivalent minimal deterministic automaton by merging equivalent states~\cite{Hop71,Moo56}.
	Notice that two equivalent minimal deterministic automata are isomorphic.

    Kleene's Theorem~\cite{Kle56} asserts that the set of the languages specified by regular expressions is the same as the set of languages recognized by finite automata.
The conversion of regular expressions into automata has been deeply studied, \emph{e.g.} by Glushkov~\cite{Glu61}.
To differentiate each occurence of the same symbol in a regular expression, 
a \emph{marking} of all the symbols of the alphabet is performed by indexing them with their relative position in the expression. 
The marking of a regular expression $E$ produces a \emph{marked regular expression} denoted by $E^{\sharp}$ over the alphabet of indexed symbols denoted by $\Pi_E$ where each indexed symbol occurs at most once in $E^{\sharp}$.
	The reverse of marking is the \emph{dropping} of subscripts, denoted by $^{\natural}$, such that if $x \in \Pi_E$ and $x = a_k$, then $x^{\natural} = a$.
	It is then extended to marked regular expressions such that $(E^{\sharp})^{\natural} = E$.
	
	Let $E$ be a regular expression over an alphabet $\Sigma$. The following functions are defined:
\begin{itemize}
	\item $\mathrm{Null}(E) = \{\varepsilon\}$ if $\varepsilon \in \mathrm{L}(E)$, $\emptyset$ otherwise
	\item $\mathrm{First}(E) = \{x \in \Sigma \mid \exists w \in \Sigma^*,\ xw \in \mathrm{L}(E)\}$
	\item $\mathrm{Last}(E) = \{x \in \Sigma \mid \exists w \in \Sigma^*,\ wx \in \mathrm{L}(E)\}$
	\item $\mathrm{Follow}(E, x) = \{y \in \Sigma \mid \exists u,v \in \Sigma^*,\ uxyv \in \mathrm{L}(E)\}$, $\forall x \in \Sigma$
\end{itemize}

	From these functions, an automaton recognizing $\mathrm{L}(E)$ can be computed:

\begin{definition}\label{def:Glushkov}
	The \emph{Glushkov automaton of a regular expression $E$ over an alphabet $\Sigma$} is denoted by $G_E = (\Sigma, Q_E, I_E, F_E, \delta_E)$ with:
	\begin{itemize}
		\item $Q_E = \Pi_E \cup \{i\}$
		\item $I_E = \{i\}$
		\item $F_E = \mathrm{Last}(E^{\sharp}) \cup \{i\}$ if $\mathrm{Null}(E^{\sharp}) = \{\varepsilon\}$, $\mathrm{Last}(E^{\sharp})$ otherwise
		\item  
		  \begin{tabular}[t]{l@{\ }l}
		  $\delta_E=$ & $\{(x, a, y) \in \Pi_E \times \Sigma \times \Pi_E \mid y \in \mathrm{Follow}(E^{\sharp}, x) \wedge a = y^{\natural}\}$\\
		  & $  \cup \{(i, a, y) \in \{i\} \times \Sigma \times \Pi_E \mid y \in \mathrm{First}(E^{\sharp}) \wedge a = y^{\natural}\}$
		  \end{tabular}
	\end{itemize}
\end{definition}
Finally, an automaton is a \emph{Glushkov automaton} if it is the Glushkov automaton of a regular expression $E$.

\begin{example}
	Let $E = (a + b)^*a + \varepsilon$. Then $E^{\sharp} = (a_1 + b_2)^*a_3 + \varepsilon$ with $\Pi_E = \{a_1, b_2, a_3\}$, and $G_E$ is given in Figure~\ref{ExempleGlushkov}.
\end{example}

\begin{figure}[H]
	\centering
	
	\begin{tikzpicture}
		\node[state, accepting, initial] (i) {$i$};
	    \node[state, accepting, right of=i] (3) {$a_3$};
    	\node[state, above right of=3] (1) {$a_1$};
		\node[state, below right of=3] (2) {$b_2$};
	    \path[->]
    		(i) edge [bend left=45] node {$a$} (1)
    		(i) edge node {$a$} (3)
   			(i) edge [swap, bend right=45] node {$b$} (2)
   		    (1) edge [loop above] node {$a$} ()
        	(1) edge [bend left=10] node {$b$} (2)
			(1) edge [swap] node {$a$} (3)
   		    (2) edge [loop below] node {$b$} ()
        	(2) edge [bend left=10] node {$a$} (1)
   		    (2) edge node {$a$} (3)
		;
	\end{tikzpicture}

  \caption{The Glushkov automaton $G_E$ of $E = (a + b)^*a + \varepsilon$}
  \label{ExempleGlushkov}
\end{figure}
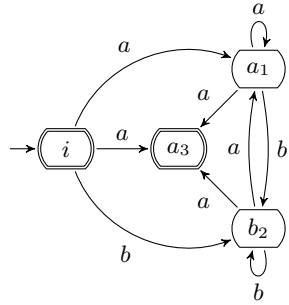

\subsection{One-Unambiguous Regular Languages}
	We present the notion of one-unambiguity introduced in~\cite{BW98}.

\begin{definition}
	A regular expression E is \emph{one-unambiguous} if $G_E$ is deterministic.
	A regular language is \emph{one-unambiguous} if it is specified by some one-unambiguous regular expression.
\end{definition}

	Brüggemann-Klein and Wood showed that the one-unambiguity of a regular language is structurally decidable over its minimal DFA. This decision procedure is related to the orbits of the underlying graph and to their links with the remaining parts:
	An automaton has the \emph{orbit property} if all the out-gates of each orbit have identical connections to the outside.
	More formally:
	
\begin{definition}
	An automaton $A=(\Sigma,Q,I,F,\delta)$ has the orbit property if, for any orbit $O$ in $\mathcal{O}_A$, for any two states $(p,q)$ in $\mathrm{G_{out}}(O)$, the two following conditions are satisfied:
	\begin{itemize}
		\item $p \in F \Longrightarrow q \in F$,
		\item $\forall r \in (Q \setminus O), \forall a \in \Sigma, r \in \delta(p,a) \Longrightarrow r \in \delta(q,a)$.
	\end{itemize}
\end{definition}
	Let $q \in Q$ be a state. 
	The \emph{orbit of a state $q$}, denoted by $\mathrm{O}(q)$ is the orbit to which $q$ belongs.
	The \emph{orbit automaton} $A_q$ \emph{of the state $q$ in $A$} is the automaton obtained by restricting the states and the transitions of $A$ to $\mathrm{O}(q)$ with initial state $q$ and final states $\mathrm{G_{out}}(\mathrm{O}(q))$. For any state $q \in Q$, the languages $\mathrm{L}(A_q)$ are called the \emph{orbit languages of $A$}.
	A symbol $a \in \Sigma$ is \emph{$A$-consistent} if there exists a state $q_a \in Q$ such that all final states of $A$ have a transition labelled by $a$ to $q_a$. 
	A set $S$ of symbols is $A$-consistent if each symbol in $S$ is $A$-consistent.
The \emph{$S$-cut} $A_S$ of $A$ is constructed from $A$ by removing, for each $a \in S$, all transitions labelled by $a$ that leave a final state of $A$.
	All these notions can be used to characterize one-unambiguous regular languages:

\begin{theorem}[\cite{BW98}]\label{th:1NA}
	Let $M$ be a minimal deterministic automaton and $S$ be a $M$-consistent set of symbols. 
	Then, $\mathrm{L}(M)$ is one-unambiguous if and only if:
	\begin{enumerate}
		\item the $S$-cut $M_S$ of $M$ has the orbit property
		\item all orbit languages of $M_S$ are one-unambiguous.
	\end{enumerate}
	Furthermore, if $M$ consists of a single non-trivial orbit and $\mathrm{L}(M)$ is one-unambi\-guous, then $M$ has at least one $M$-consistent symbol.
\end{theorem}
	This theorem suggests an inductive algorithm to decide, given a minimal deterministic automaton $M$ whether $\mathrm{L}(M)$ is one-unambiguous: 
	the \emph{BKW test}.
Furthermore, the theorem defines a sufficient condition over non-minimal deterministic automaton:

\begin{lemma}[\cite{BW98}]\label{lm:1NA}
	Let $A$ be a deterministic automaton and $M$ be its equivalent minimal deterministic automaton.
	\begin{enumerate}
		\item If $A$ has the orbit property, then so does $M$
		\item If all orbit languages of $A$ are one-unambiguous, then so are all orbit languages of $M$.
	\end{enumerate}
\end{lemma}
	Consequently, 
	the BKW test is extended to deterministic automata which are not minimal.
	Reinterpreting the results in~\cite{BW98}, it can be shown that
	
\begin{lemma}\label{lm:GlushkovBKW}
	  The Glushkov automaton of a one-unambiguous regular expression passes the BKW test.
\end{lemma}

\subsection{Lookahead Deterministic Regular Languages}

	We present the notion of lookahead determinism introduced in~\cite{HW08}.
	The basic idea is that the reading of the next $k$ symbols of the input word allows one to know the next position in the expression or in the automaton.
	
\begin{definition}
	An automaton $A = (\Sigma, Q, I, F, \delta)$ is \emph{$k$-lookahead deterministic} if the following conditions hold:
	\begin{itemize}
		\item $|I| = 1$
		\item $\forall t_1 = (p, a, q_1), t_2 = (p, b, q_2) \in \delta, (t_1 \neq t_2) \Longrightarrow (a \neq b) \vee (\forall w \in \Sigma^{k-1}, \delta(q_1, w) = \emptyset \vee \delta(q_2, w) = \emptyset)$.
	\end{itemize}
\end{definition}

\begin{definition}
	A regular expression $E$ is $k$-lookahead deterministic if $G_E$ is $k$-lookahead deterministic.
	A regular language is $k$-lookahead deterministic if it is specified by some $k$-lookahead deterministic regular expressions.
\end{definition}

	Since a $1$-lookahead deterministic automaton is deterministic, the family of $1$-lookahead deterministic language is the same as the family of one-unambiguous language.

\begin{example}
	Let $E = b^*a(b^*a)^*(a+b)$, $G_E$ is given in Figure~\ref{fg:Glushkov2LA}.
	Notice that the states $a_2$ and $a_4$ admit two successors by $a$ and $b$, but since $\mathrm{L}_{a_5}(G_E) = \mathrm{L}_{b_6}(G_E) = \{\varepsilon\}$, then $G_E$ and $E$ are $2$-lookahead deterministic. 
\end{example}

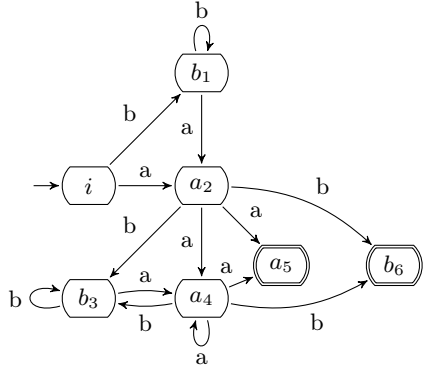
\begin{figure}[H]
	\centering
	
	\begin{tikzpicture}
		\node[state, initial] (i) {$i$};
	    \node[state, right of=i] (2) {$a_2$};
    	\node[state, above of=2] (1) {$b_1$};
		\node[state, below of=2] (4) {$a_4$};
		\node[state, left of=4] (3) {$b_3$};
		\node[state, accepting, below right of=2] (5) {$a_5$};
		\node[state, accepting, right of=5] (6) {$b_6$};		
	    \path[->]
    			(i) edge node {b} (1)
			(i) edge node {a} (2)
   		    (1) edge [loop above] node {b} ()
        		(1) edge [swap] node {a} (2)
			(2) edge [swap] node {b} (3)
   		    (2) edge [swap] node {a} (4)
   		    (2) edge node {a} (5)
   		    (2) edge [bend left=20] node {b} (6)
        		(3) edge [loop left] node {b} ()
			(3) edge [bend left=10] node {a} (4)
			(4) edge [bend left=10] node {b} (3)
        		(4) edge [loop below] node {a} ()
   		    (4) edge node {a} (5)
   		    (4) edge [swap, bend right=20] node {b} (6)
		;
	\end{tikzpicture}

	\caption{The $2$-lookahead deterministic Glushkov automaton $G_E$}
	\label{fg:Glushkov2LA}
\end{figure}

	It has been proved in~\cite{HW08} that the language $\mathrm{L}(b^*a(b^*a)^*(a+b))$ is not one-unambiguous.
	Thus, one-unambiguous regular languages are a proper subfamily of $k$-lookahead deterministic regular languages.

\subsection{Block Deterministic Regular Languages}

	We present the notion of block determinism introduced in~\cite{GMW01}.
	
	Let $\Sigma$ be an alphabet and $k$ be an integer. The \emph{set of blocks} $B_{\Sigma, k}$ is the set $\{w \mid w \in \Sigma^* \wedge 1 \leq |w| \leq k\}$. The notions of regular expression and automaton can be extended to ones over set of blocks.
	Let $E$ be a regular expression over $\Gamma$ and $A = (\Gamma, Q, I, F, \delta)$ be an automaton. Let $\Sigma$ be an alphabet and $k$ be an integer, if $\Gamma \subset B_{\Sigma, k}$ then $E$ and $A$ are $(\Sigma, k)$\emph{-block}.
	And since $\Gamma \subset B_{\Sigma, k} \subset \Sigma^*$, a language over $\Gamma$ is also a language over $\Sigma$.
	To distinguish blocks as syntactic components in a regular expression, we write them between square brackets.
Those are omitted for one letter blocks.

	Since $\Sigma = B_{\Sigma, 1}$, regular expressions and automata can be considered as ones over a set of blocks.
Moreover, the blocks can be treated as single symbols, as we do when we refer to the elements of an alphabet.
With this assumption,  the marking of block regular expressions induces the construction of a Glushkov automaton from a block regular expression, and the usual automaton transformations such as determinization and minimization can be easily performed.

\begin{example}
	Let $E = [aa]^*([ab]b + ba)b^*$. Then $E^{\sharp} = [aa]_1^*([ab]_2b_3 + b_4a_5)b_6^*$, and $G_E$ is given in Figure~\ref{fg:GlushkovBloc}.
\end{example}

\begin{figure}[H]
	\centerline{
		\begin{tikzpicture}
			\node[state, initial] (i) {$i$};
		    \node[state, below of=i] (1) {$[aa]_1$};
		    \node[state, left of=1] (4) {$b_4$};
		    \node[state, right of=1] (2) {$[ab]_2$};
		    \node[state, accepting, below of=2] (3) {$b_3$};
		    \node[state, accepting, below of=4] (5) {$b_5$};
	    	\node[state, accepting, below of=1] (6) {$b_6$};
		    \path[->]
	    		(i) edge node {$aa$} (1)
	    		(i) edge [swap] node {$b$} (4)
	       		(i) edge [bend left=45] node {$ab$} (2)
	        	(1) edge [loop below] node {$aa$} ()
	    		(1) edge node {$b$} (4)
	   		    (1) edge [swap] node {$ab$} (2)
	   		    (2) edge node {$b$} (3)
	  		    (3) edge node {$b$} (6)
	   		    (4) edge [swap] node {$a$} (5)
	   		    (5) edge [swap] node {$b$} (6)
	   		    (6) edge [loop below] node {$b$} ()
			;
		\end{tikzpicture}
	}
  \caption{The $(\{a, b\}, 2)$-block Glushkov automaton $G_E$}
  \label{fg:GlushkovBloc}
\end{figure}
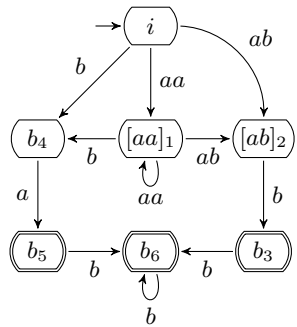

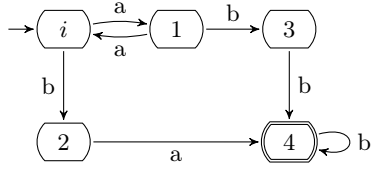
\begin{figure}[H]
	\centering
	
	\begin{tikzpicture}
		\node[state, initial] (i) {$i$};
	    \node[state, right of=i] (1) {1};
	    \node[state, below of=i] (2) {2};
	    \node[state, right of=1] (3) {3};
	    \node[state, accepting, below of=3] (4) {4};
	    \path[->]
	    		(i) edge [bend left=10] node {a} (1)
    		    (i) edge [swap] node {b} (2)
    		    (1) edge [bend left=10] node {a} (i)
    		    (1) edge node {b} (3)
    		    (2) edge [swap] node {a} (4)
   		    (3) edge node {b} (4)
    		    (4) edge [loop right] node {b} ()
		;
	\end{tikzpicture}

  \caption{The minimal DFA of $\mathrm{L}(E)$}
  \label{fg:mdfa_kbd}
\end{figure}

The notion of determinism can also be extended to block determinism as follows:

\begin{definition}
	An automaton $A = (\Gamma, Q, I, F, \delta)$ is $k$\emph{-block deterministic} if the following conditions hold:
		\begin{itemize}
			\item there exists an alphabet $\Sigma$ such that $A$ is $(\Sigma, k)$-block,
			\item $|I| = 1$,
			\item $\forall t_1 = (p, b_1, q_1), t_2 = (p, b_2, q_2) \in \delta, (t_1 \neq t_2) \Longrightarrow (b_1 \notin \mathrm{Pref}(b_2))$.
		\end{itemize}
\end{definition}	

Finally, the block determinism of a Glushkov automaton can be used to extend the block determinism to block expression:

\begin{definition}
	A block regular expression $E$ is $k$-block deterministic if $G_E$ is $k$-block deterministic.
	A regular language is $k$-block deterministic if it is specified by some $k$-block deterministic regular expressions.
\end{definition}

Since a $1$-block deterministic automaton is a deterministic automaton, the family of $1$-block deterministic language is the same as the family of one-unambi\-guous language.

\begin{example}
	Since the Glushkov automaton in Figure~\ref{fg:GlushkovBloc} is $2$-block deterministic, $\mathrm{L}([aa]^*([ab]b + ba)b^*)$ is $2$-block deterministic.
\end{example}
	
	Let $A=(\Sigma,Q,I,F,\delta)$ be an automaton and $\Gamma$ be a set. Then the automaton $B=(\Gamma,Q,I,F,\delta')$ is \emph{an alphabetic image of} $A$ if there exists an injection $\phi$ from $\Sigma$ to $\Gamma$ such that $\delta'=\{(p,\phi(a),q)\mid (p,a,q) \in\delta\}$. In this case, we set $B=\phi(A)$.	
Caron and Ziadi showed in~\cite{CZ97} that an automaton is a Glushkov one if and only if the two conditions hold:
\begin{itemize}
  \item it is homogeneous (for any state $q$, for any two transitions $(p,a,q)$ and $(r,b,q)$, the symbols $a$ and $b$ are the same);
  \item it satisfies some structural properties over the transition structure.
\end{itemize}
One can check that any injection $\phi$ from $\Sigma$ to $\Gamma$ preserves such conditions, since
the alphabetical image preserves the transition structure by only changing the symbol labeling a transition.
Therefore 
	\begin{lemma}\label{lm:ImageGlushkov}
	  The alphabetic image of an automaton $A$ is a Glushkov automaton if and only if $A$ is a Glushkov automaton. 
	\end{lemma}	
	Let us show that the BKW test can be used to characterize the $k$-block determinism of a regular language:
\begin{theorem}\label{th:KBD}
	A regular language $L$ is $k$-block deterministic if and only if it is recognized by a $k$-block deterministic automaton $K$ such that $K$ is the alphabetic image of a deterministic automaton which passes the BKW test.
\end{theorem}
\begin{proof}
  Let us show the double implication.
  \begin{enumerate}
	\item Let $L$ be a $k$-block deterministic regular language over $\Sigma$.
	Then there exists a $k$-block deterministic Glushkov automaton $K = (B_{\Sigma, k}, Q, \{i\}, F, \delta_K)$ that recognizes $L$.
	Let $\Pi = \{[b] \mid b \in B_{\Sigma, k}\}$ be an alphabet, $\varphi : \Pi \rightarrow B_{\Sigma, k}$ be
	the bijection such that for every $[b] \in \Pi, \varphi([b]) = b$. Let $A = (\Pi, Q, \{i\}, F, \delta_A)$ be a Glushkov automaton such that $K = \varphi(A)$.
	Let us suppose that $A$ is not deterministic.
	Then, there exist two transitions $(p, a, q), (p, a, r) \in \delta_A$ such that $q \neq r$.
	Thus, $(p, \varphi(a), q), (p, \varphi(a), r) \in \delta_K$, which contradicts the fact that $K$ is $k$-block deterministic.
	So, $A$ is a deterministic Glushkov automaton, and therefore passes the BKW test following Lemma~\ref{lm:GlushkovBKW}.
	
	\item Let $A = (\Pi, Q_A, \{i_A\}, F_A, \delta_A)$ be a deterministic automaton which passes the BKW test, $K = \{\Gamma, Q_A, \{i_A\}, F_A, \delta_K)$ be a $k$-block deterministic automaton, and $\varphi : \Pi \rightarrow \Gamma$ be an injection such that $K = \varphi(A)$.
	Now, $\varphi : \Pi \rightarrow \Gamma$ is extended into the morphism $\varphi : \Pi^* \rightarrow \Gamma^*$ such that for every letter $a \in \Pi$ and every word $w \in \Pi^*$ we have $\varphi(a \cdot w) = \varphi(a) \cdot \varphi(w)$ and $\varphi(\varepsilon) = \varepsilon$.
	In this case, $\mathrm{L}(K) = \varphi(\mathrm{L}(A))$.
	Since $A$ passes the BKW test, there exists an equivalent deterministic Glushkov automaton $G = (\Pi, Q_G, \{i_G\}, F_G, \delta_G)$. Following Lemma~\ref{lm:ImageGlushkov}, there also exists a Glushkov automaton $H = (\Gamma, Q_G, \{i_G\}, F_G, \delta_H)$ such that $H = \varphi(G)$ and $\mathrm{L}(H) = \varphi(\mathrm{L}(G))$.
	Since $A$ and $G$ are equivalent deterministic automata,
	$\varphi(\mathrm{L}(G)) = \varphi(\mathrm{L}(A))$.
	And so $\mathrm{L}(H) = \mathrm{L}(K)$.
	Let us suppose that $H$ is not $k$-block deterministic, then there exist two transitions $(p_H, \varphi(a), q_H), (p_H, \varphi(b), r_H) \in \delta_H$ such that either $(\varphi(a) = \varphi(b)) \wedge (q_H \neq r_H)$ or $(\varphi(a) \neq \varphi(b)) \wedge (\varphi(a) \in \mathrm{Pref}(\varphi(b)))$.
	By definition, $(p_H, a, q_H),$ $(p_H, b, r_H) \in \delta_G$.
	But since $G$ and $A$ are equivalent deterministic automata, there exist two transitions $(p_A, a, q_A)$, $(p_A, b, r_A) \in \delta_A$, and by definition, $(p_A, \varphi(a), q_A)$, $(p_A, \varphi(b), r_A) \in \delta_K$.
	Let us suppose that $(\varphi(a) = \varphi(b)) \wedge (q_H \neq r_H)$.
	Since $\varphi$ is an injection, $(a = b) \wedge (q_H \neq r_H)$, which contradicts the fact that $G$ is deterministic.
	So let us suppose that $(\varphi(a) \neq \varphi(b)) \wedge (\varphi(a) \in \mathrm{Pref}(\varphi(b)))$, it contradicts the fact that $K$ is $k$-block deterministic.
	Therefore, $H$ is a $k$-block deterministic Glushkov automaton, and $\mathrm{L}(K)$ is $k$-block deterministic.
	\end{enumerate}
\end{proof}

	It has been proved that one-unambiguous regular languages are a proper subfamily of $k$-block deterministic regular languages.
	As an example, the language $\mathrm{L}([aa]^*([ab]b + ba)b^*)$ is $2$-block deterministic but not one-unambiguous since its minimal deterministic automaton given in Figure \ref{fg:mdfa_kbd} does not pass the BKW test.
Therefore one can wonder whether there exists an infinite hierarchy in $k$-block deterministic regular languages regarding $k$.
That has been achieved by Han and Wood~\cite{HW08}, but with an invalid assumption.

\section{Previous Results on Block Deterministic Languages}\label{se:PrevRes}

	In~\cite{GMW01}, a method is presented for creating from a block automaton an equivalent block automaton with larger blocks by eliminating states while preserving the right language of every other states.
	
	Let $A = (\Gamma, Q, I, F, \delta)$ be a block automaton.
	The \emph{state elimination of $q$ in $A$} creates a new block automaton, denoted by $\mathcal{S}(A,q)$, computed as follows: first, the state $q$ and all transitions going in and out of it are removed; second, for every two transitions $(r, u, q)$ and $(q, v, s)$ in $\delta$, the transition $(r, uv, s)$ is added.
	This transformation is illustrated in Figure~\ref{fg:StateElim}.

\begin{figure}[H]
	\begin{minipage}[b]{.48\linewidth}
		\centering
		
		\begin{tikzpicture}
			\node[state] (q) {$q$};
			\node[state, above left of=q] (r1) {$r_1$};
		    \node[state, below left of=q] (r2) {$r_2$};
    		\node[state, above right of=q] (s1) {$s_1$};
    		\node[state, below right of=q] (s2) {$s_2$};
		    \path[->]
    		    (r1) edge [swap]node {$u_1$} (q)
    		    (r2) edge node {$u_2$} (q)
    		    (q) edge [bend left=10] node {$v_1$} (s1)
	   			(q) edge node {$v_2$} (s2)
	   			(s1) edge [bend left=10] node {$w$} (q)
			;
		\end{tikzpicture}
	
	\end{minipage}
	\hfill
	\begin{minipage}[b]{.48\linewidth}
		\centering
			
		\begin{tikzpicture}	
			\node (q) {};
			\node[state, above left of=q] (r1) {$r_1$};
		    \node[state, below left of=q] (r2) {$r_2$};
    		\node[state, above right of=q] (s1) {$s_1$};
    		\node[state, below right of=q] (s2) {$s_2$};
		    \path[->]	    		
    		    (r1) edge node {$u_1v_1$} (s1)
    		    (r1) edge [swap, near start] node {$u_1v_2$} (s2)
    		    (r2) edge [near start] node {$u_2v_1$} (s1)
	   			(r2) edge [swap] node {$u_2v_2$} (s2)
	   			(s1) edge [loop above] node {$wv_1$} ()
    		    (s1) edge node {$wv_2$} (s2)	   			
			;
		\end{tikzpicture}

	\end{minipage}
	\caption{The state elimination of the state $q$}
	\label{fg:StateElim}
\end{figure}
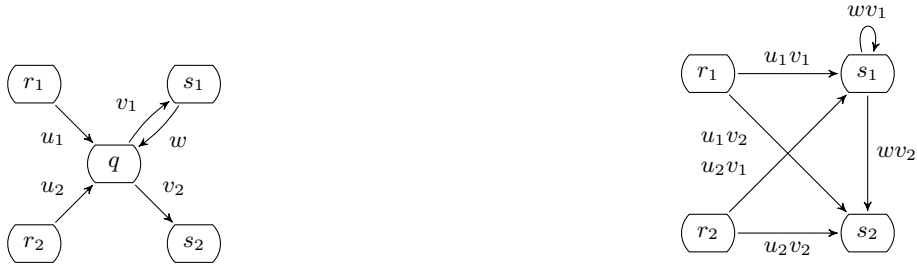

\begin{definition}
	Let $A = (\Gamma, Q, I, F, \delta)$ be a block automaton.
	A state $q \in Q$ satisfies the \emph{state elimination precondition} if it is neither an initial state nor a final state and it has no self-loop.
\end{definition}

	The state elimination is extended to a set $S \subset Q$ of states if every state in $S$ satisfies the state elimination precondition, and the subgraph induced by $S$ is acyclic. In this case, we can eliminate the states in $S$ in any order. Giammarresi \emph{et al.} \cite{GMW01} suggest that state elimination is sufficient to decide the $k$-block determinism of a regular language.

\begin{lemma}[\cite{GMW01,HW08}]\label{le:GMWH}
	Let $M$ be a minimal deterministic automaton of a $k$-block deterministic regular language. We can transform $M$ to a $k$-block deterministic automaton that satisfies the orbit property using state elimination.
\end{lemma}

Using this lemma, Han and Wood stated that:

\begin{theorem}[\cite{HW08}]\label{th:HW}
	There is a proper hierarchy in $k$-block deterministic regular languages.
\end{theorem}
\begin{proof}
	Han and Wood exhibited the family of languages $L_k$ specified by regular expressions $E_k = ([a^k])^*([a^{k-1}b]b + ba)b^*$ whose minimal deterministic automata $M_k$ are represented in Figure~\ref{fg:LEk}.
	Following Lemma~\ref{le:GMWH}, there is no other choice but to eliminate states $q_1$ to $q_{k-1}$, in any order, to have the orbit property.
	Thus, $L_k$ is $k$-block deterministic and not $(k-1)$-block deterministic.
\end{proof}

\begin{figure}[H]
	\begin{minipage}[b]{.55\linewidth}
		\centering
		
		\begin{tikzpicture}
			\tikzstyle{every state}=[rounded rectangle, rounded rectangle arc length=90, inner sep=0.5mm, minimum size=5mm]		
		
			\node[state, initial] (qk) {$q_k$};
			\node[state, below of=qk] (1) {$1$};
		    \node[state, above of=qk] (qk1) {$q_{k-1}$};
    		\node[state, right of=qk1] (qk2) {$q_{k-2}$};
    		\node[state, right of=qk2, xshift=0.5cm] (q3) {$q_3$};
    		\node[state, right of=q3] (q2) {$q_2$};
    		\node[state, below of=q2] (q1) {$q_1$};
    		\node[state, below of=q1] (2) {$2$};
    		\node[state, accepting, left of=2] (3) {$3$};
		    \path[->]
    		    (qk) edge node {$a$} (qk1)
    		    (qk) edge [swap] node {$b$} (1)
    		    (qk1) edge node {$a$} (qk2)
    		    (q3) edge node {$a$} (q2)
    		    (q2) edge node {$a$} (q1)
    		    (q1) edge [swap] node {$a$} (qk)
	   			(q1) edge node {$b$} (2)
	   			(2) edge node {$b$} (3)
	   			(1) edge [swap] node {$a$} (3)
	   			(3) edge [loop above] node {$b$} ()
			;
			\draw[densely dashed] (qk2) -- (q3);
		\end{tikzpicture}
	
	\end{minipage}
	\hfill
	\begin{minipage}[b]{.44\linewidth}
		\centering
			
		\begin{tikzpicture}
			\tikzstyle{every state}=[rounded rectangle, rounded rectangle arc length=90, inner sep=0.5mm, minimum size=5mm]		
		
			\node[state, initial] (qk) {$q_k$};
			\node[state, below of=qk] (1) {$1$};
    		\node[state, accepting, right of=1] (3) {$3$};
    		\node[state, right of=3] (2) {$2$};
		    \path[->]
    		    (qk) edge [loop above] node {$a^k$} ()
    		    (qk) edge [swap] node {$b$} (1)
    		    (qk) edge [bend left=35] node {$a^{k-1}b$} (2)
	   			(2) edge node {$b$} (3)
	   			(1) edge [swap] node {$a$} (3)
	   			(3) edge [loop above] node {$b$} ()
			;
		\end{tikzpicture}

	\end{minipage}
		
	\caption{
	The minimal deterministic automaton $M_k$ and its equivalent $k$-block de\-ter\-mi\-nistic automaton after having eliminated states $q_1$ to $q_{k-1}$
	}
	\label{fg:LEk}
\end{figure}
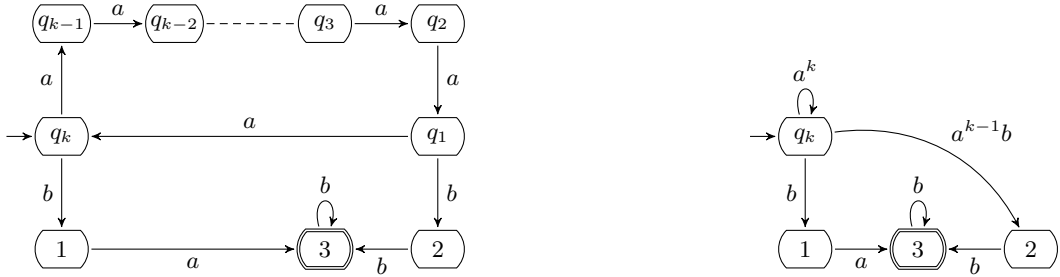

\section{A Witness for the Infinite Hierarchy of Block Deterministic Languages}\label{se:kbd_witness}

In this section, we exhibit a counter-example for Lemma~\ref{le:GMWH}. We can find a $k$-block deterministic language with a minimal deterministic automaton from which we cannot get any $k$-block deterministic automaton that satisfies the orbit property. 
	In Figure~\ref{mp:dfam}, the leftmost automaton is minimal and none of its states can be eliminated.
	However, by applying standardization, we create an equivalent deterministic automaton from which we can eliminate the state $i$ to get the rightmost equivalent $2$-block deterministic automaton.
	
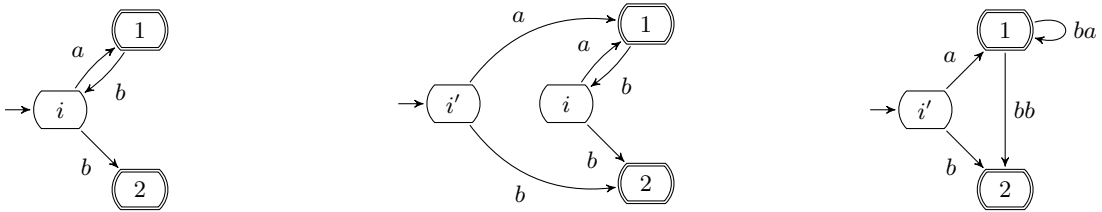
\begin{figure}[H]
	\begin{minipage}[b]{.32\linewidth}		
		\centerline{
			\begin{tikzpicture}
				\node[state, initial] (i) {$i$};
			    \node[state, accepting, above right of=i] (1) {$1$};
		    	\node[state, accepting, below right of=i] (2) {$2$};
			    \path[->]
		    		(i) edge [bend left=10] node {$a$} (1)
		    		(i) edge [swap] node {$b$} (2)
		   			(1) edge [bend left=10] node {$b$} (i)
				;
			\end{tikzpicture}
		}
	\end{minipage}
	\hfill
	\begin{minipage}[b]{.32\linewidth}		
		\centerline{
			\begin{tikzpicture}
				\node[state, initial] (i') {$i'$};
				\node[state, right of=i'] (i) {$i$};
			    \node[state, accepting, above right of=i] (1) {$1$};
		    	\node[state, accepting, below right of=i] (2) {$2$};
			    \path[->]
			    	(i') edge [bend left=30] node {$a$} (1)
			    	(i') edge [swap, bend right=30] node {$b$} (2)
		    		(i) edge [bend left=10] node {$a$} (1)
		    		(i) edge [swap] node {$b$} (2)
		   			(1) edge [bend left=10] node {$b$} (i)
				;
			\end{tikzpicture}
		}
	\end{minipage}
	\hfill
	\begin{minipage}[b]{.32\linewidth}	
		\centerline{	
			\begin{tikzpicture}
				\node[state, initial] (i) {$i'$};
			    \node[state, accepting, above right of=i] (1) {$1$};
		    	\node[state, accepting, below right of=i] (2) {$2$};
			    \path[->]
		    	    (i) edge  node {$a$} (1)
		    		(i) edge [swap] node {$b$} (2)
		   		    (1) edge [loop right] node {$ba$} ()
		        	(1) edge node {$bb$} (2)
				;
			\end{tikzpicture}
		}
	\end{minipage}
	  	\caption{The counter-example}
	  	\label{mp:dfam}
\end{figure}	

This clearly shows that the only action of state elimination is not enough to decide whether a language is $k$-block deterministic.
Using this operation, we show that:

\begin{proposition}
	$\forall k \in \mathbb{N} \setminus \{0\}$, the language $L_k$ is $2$-block deterministic.
\end{proposition}
\begin{proof}
	As shown in Figure~\ref{fg:ContreEx}, we can always standardize $M_k$, proceed to the state elimination of $q_k$ and get a $2$-block deterministic automaton which respects the conditions stated in Theorem~\ref{th:KBD}.
Thus, $L_k$ is $2$-block deterministic and is specified by the regular expressions $F_k = (a^{k-1}([aa]a^{k-2})^*([ab]a + bb) + ba)b^*$.
\end{proof}

\begin{figure}[H]
	\begin{minipage}[b]{.48\linewidth}
		\centering	
		\begin{tikzpicture}[transform shape, scale=0.9]
			\tikzstyle{every state}=[rounded rectangle, rounded rectangle arc length=90, inner sep=0.5mm, minimum size=5mm]	
			\node[state, initial] (i) {$i$};
			\node[state, right of=i] (qk) {$q_k$};
			\node[state, below of=i] (1) {$1$};
		    \node[state, above of=i] (qk1) {$q_{k-1}$};
    		\node[state, right of=qk1] (qk2) {$q_{k-2}$};
    		\node[state, right of=qk2, xshift=0.5cm] (q3) {$q_3$};
    		\node[state, right of=q3] (q2) {$q_2$};
    		\node[state, below of=q2] (q1) {$q_1$};
    		\node[state, below of=q1] (2) {$2$};
    		\node[state, accepting, left of=2] (3) {$3$};
		    \path[->]
	    		(i) edge node {$a$} (qk1)
    		    (i) edge [swap] node {$b$} (1)
    		    (qk) edge [swap] node {$a$} (qk1)
    		    (qk) edge node {$b$} (1)
    		    (qk1) edge node {$a$} (qk2)
    		    (q3) edge node {$a$} (q2)
    		    (q2) edge node {$a$} (q1)
    		    (q1) edge [swap] node {$a$} (qk)
	   			(q1) edge node {$b$} (2)
	   			(2) edge node {$b$} (3)
	   			(1) edge [swap] node {$a$} (3)
	   			(3) edge [loop above] node {$b$} ()
			;
			\draw[densely dashed] (qk2) -- (q3);
		\end{tikzpicture}	
	\end{minipage}
	\hfill
	\begin{minipage}[b]{.48\linewidth}
		\centering	
		\begin{tikzpicture}[transform shape, scale=0.9]
			\tikzstyle{every state}=[rounded rectangle, rounded rectangle arc length=90, inner sep=0.5mm, minimum size=5mm]			
			\node[state, initial] (i) {$i$};
			\node[state, below of=i] (1) {$1$};
		    \node[state, above of=i] (qk1) {$q_{k-1}$};
    		\node[state, right of=qk1] (qk2) {$q_{k-2}$};
    		\node[state, right of=qk2, xshift=0.5cm] (qa3) {$q_3$};
    		\node[state, right of=q3] (qa2) {$q_2$};
    		\node[state, below of=q2] (qa1) {$q_1$};
    		\node[state, below of=q1] (2) {$2$};
    		\node[state, accepting, left of=2] (3) {$3$};
	   		\path[->]
    		    (i) edge node {$a$} (qk1)
    		    (i) edge [swap] node {$b$} (1)
    		    (qk1) edge node {$a$} (qk2)
    		    (q3) edge node {$a$} (q2)
    		    (q2) edge node {$a$} (q1)
		    	(q1) edge [bend left=10, swap] node {$aa$} (qk1)
    		    (q1) edge [bend right=10] node {$ab$} (1)
   				(q1) edge node {$b$} (2)
	   			(2) edge node {$b$} (3)
	   			(1) edge [swap] node {$a$} (3)
   				(3) edge [loop above] node {$b$} ()
			;
			\draw[densely dashed] (qk2) -- (q3);
		\end{tikzpicture}
	\end{minipage}	
	\caption{The standardization of $M_k$ followed by the state elimination of $q_k$}
	\label{fg:ContreEx}
\end{figure}

	However, Theorem~\ref{th:HW} is still correct since we can give proper details about the proof with our own parameterized family of languages.
Let $k \in \mathbb{N} \setminus \{0\}$ be an integer and $A_k = (\Sigma, Q_k, I_k, F_k, \delta_k)$ be the automaton (given in Figure~\ref{fg:Ak}) such that:
	\begin{itemize}
		\item $\Sigma = \{a, b, c\}$
		\item $Q_k = \{f\} \cup \{\alpha_j, \beta_j \mid 1 \leq j \leq k\}$
		\item $I_k = \{\beta_k\}$
		\item $F_k = \{f\} \cup \{\alpha_k, \beta_k\}$
		\item $\delta_k = \Delta_k \cup \Gamma_k$ with:
			\begin{itemize}
				\item $\Delta_k = \{(\beta_k, a, \alpha_k), (\beta_1, b, f), (\alpha_k, a, \alpha_k), (\alpha_1, b, f), (\alpha_1, c, \beta_k)\}$
				\item $\Gamma_k = \{(\alpha_j, b, \alpha_{j-1}), (\beta_j, b, \beta_{j-1}) \mid 2 \leq j \leq k\}$
			\end{itemize}
	\end{itemize}	
	
\begin{figure}[H]
	\centering
	
	\begin{tikzpicture}
		\tikzstyle{every state}=[rounded rectangle, rounded rectangle arc length=90, inner sep=0.5mm, minimum size=5mm]
	
		\node (anchor) {};	
		\node[state, initial, accepting, above right of=anchor] (bk) {$\beta_k$};
		\node[state, right of=bk] (bk1) {$\beta_{k-1}$};
		\node[state, right of=bk1] (bk2) {$\beta_{k-2}$};
		\node[state, right of=bk2, xshift=0.5cm] (b2) {$\beta_2$};
		\node[state, right of=b2] (b1) {$\beta_1$};
	    \node[state, accepting, below right of=anchor] (ak) {$\alpha_k$};
		\node[state, right of=ak] (ak1) {$\alpha_{k-1}$};
		\node[state, right of=ak1] (ak2) {$\alpha_{k-2}$};
		\node[state, right of=ak2, xshift=0.5cm] (a2) {$\alpha_2$};
		\node[state, right of=a2] (a1) {$\alpha_1$};
    		\node[state, accepting, below right of=b1] (f) {$f$};
	    \path[->]
		    (bk) edge node {$b$} (bk1)
		    (bk) edge [swap] node {$a$} (ak)
		    (bk1) edge node {$b$} (bk2)
		    (bk2) edge [densely dashed] node {} (b2)
	   	 	(b2) edge node {$b$} (b1)
	  	  	(b1) edge node {$b$} (f)
   			(ak) edge [loop left] node {$a$} ()
   			(ak) edge [swap] node {$b$} (ak1)
   			(ak1) edge [swap] node {$b$} (ak2)
   			(ak2) edge [densely dashed] node {} (a2)
		    	(a2) edge [swap] node {$b$} (a1)
		    	(a1) edge [swap] node {$b$} (f)
		   	(a1) edge [out=135, in=315, swap] node {$c$} (bk)
		;
	\end{tikzpicture}
	
	\caption{The deterministic automaton $A_k$}
	\label{fg:Ak}
\end{figure}	

	First of all, let us notice that the word $b^j\in \mathrm{L}(A_k)$ if and only if $j=k$.
	Thus, for all $k \neq k'$, $\mathrm{L}(A_k) \neq \mathrm{L}(A_{k'})$. Furthermore,

\begin{proposition}
	$\forall k \in \mathbb{N} \setminus \{0\}$, $\mathrm{L}(A_k)$ is $k$-block deterministic.
\end{proposition}
\begin{proof}
	By construction, for all $k$, $A_k$ is trimmed and deterministic.
	So, any automaton that we can get from eliminating states such that the state elimination precondition is respected is a block deterministic automaton. 

	For any integer $k$ in $\mathbb{N} \setminus \{0\}$, we can eliminate the set of states $\{\alpha_j, \beta_j \mid 1 \leq j \leq k - 1\}$ because none of these states are initial or final and their induced subgraph is acyclic.	
	Thus, we can get a $k$-block deterministic automaton $B_k$, such that $\mathrm{L}(B_k) = \mathrm{L}(A_k)$, shown in Figure~\ref{fg:bk}. 	
	Obviously $B_k$ respects the conditions stated in Theorem~\ref{th:KBD}, so $\mathrm{L}(A_k)$ is $k$-block deterministic. Furthermore, it can be checked that $\mathrm{L}(A_k)$ is specified by the $k$-block deterministic regular expression $(a(\varepsilon + [b^{k-1}c]))^*(\varepsilon + [b^k])$.
\end{proof}

\begin{figure}[H]
	\centering
	
	\begin{tikzpicture}
		\tikzstyle{every state}=[rounded rectangle, rounded rectangle arc length=90, inner sep=0.5mm, minimum size=5mm]
	
		\node (anchor) {};	
		\node[state, initial, accepting, above right of=anchor] (bk) {$\beta_k$};
		\node[state, accepting, below right of=anchor] (ak) {$\alpha_k$};
    	\node[state, accepting, below right of=bk] (f) {$f$};
	    \path[->]
		    (bk) edge [bend left=25] node {$b^k$} (f)
		    (bk) edge [bend left=10] node {$a$} (ak)
   			(ak) edge [loop left] node {$a$} ()
   			(ak) edge [bend right=25, swap] node {$b^k$} (f)
	    	(ak) edge [bend left=10] node {$b^{k-1}c$} (bk)
		;
	\end{tikzpicture}
	
	\caption{The $k$-block deterministic automaton $B_k$}
	\label{fg:bk}
\end{figure}		

Finally, let us show that the index cannot be reduced:

\begin{proposition}
	$\forall k \in \mathbb{N} \setminus \{0, 1\}$, $\mathrm{L}(A_k)$ is not $(k-1)$-block deterministic.
\end{proposition}
\begin{proof}
	Let $B = (B_{\Sigma, k-1}, Q_B, \{i_B\}, F_B, \delta_B)$ be a $(k-1)$-block deterministic automaton equivalent to $A_k$.
	
	We first show that there exists a non-trivial orbit $O \subset Q_B$ and two states $\alpha, \beta \in O$ such that $\mathrm{L}_{\alpha}(B) = \mathrm{L}_{\alpha_k}(A_k)$ and $\mathrm{L}_{\beta}(B) = \mathrm{L}_{\beta_k}(A_k)$.
	Let us consider the following state sequences: $(\alpha_{k, j})_{j \in \mathbb{N}} \subset F_B$ and $(\beta_{k, j})_{j \in \mathbb{N}} \subset F_B$, such that $\beta_{k, 0} = i_B$, $\delta_B(\beta_{k, j}, a) = \alpha_{k, j}$ and $\delta_B(\alpha_{k, j}, b^{k-1}c) = \beta_{k, j+1}$.
	It follows that $\delta_B(i_B, (ab^{k-1}c)^j) = \beta_{k, j}$ and $\delta_B(i_B, (ab^{k-1}c)^{j}a) = \alpha_{k, j}$.
	Notice that the existence of $\alpha_{k, j}$ and $\beta_{k, j}$ is ensured by the fact that $\mathrm{L}(B) = \mathrm{L}(A_k)$.
	Let us suppose that there exists $j \in \mathbb{N}$ such that $\mathrm{L}_{\beta_{k, j}}(B) \neq \mathrm{L}_{\beta_k}(A_k)$.
	Then there exists $w \in \Sigma^*$ such that $w \in \mathrm{L}_{\beta_{k, j}}(B) \bigtriangleup \mathrm{L}_{\beta_k}(A_k)$, where for any two sets $X$ and $Y$, $X\bigtriangleup Y=(X\setminus Y) \cup (Y\setminus X)$.
	And since $\delta_k(\beta_k, (ab^{k-1}c)^j) = \beta_k$, $(ab^{k-1}c)^j \cdot w \in \mathrm{L}(B) \bigtriangleup \mathrm{L}(A_k)$.
	Thus, $\mathrm{L}(B) \neq \mathrm{L}(A_k)$ which is contradictory.
	So, for every $j \in \mathbb{N}$, we have $\mathrm{L}_{\beta_{k, j}}(B) = \mathrm{L}_{\beta_k}(A_k)$.
	The proof that for every $j \in \mathbb{N}$, we have $\mathrm{L}_{\alpha_{k, j}}(B) = \mathrm{L}_{\alpha_k}(A_k)$, is done in the same way.
	Now, let us suppose that for every $j \neq j' \in \mathbb{N}$, we have $\alpha_{k, j} \neq \alpha_{k, j'}$ and $\beta_{k, j} \neq \beta_{k, j'}$.
	Then $Q_B$ would be infinite, which would contradict the fact that $B$ is a finite automaton.
	So, there exist $j < j' \in \mathbb{N}$ such that $\alpha_{k, j} = \alpha_{k, j'}$ or $\beta_{k, j} = \beta_{k, j'}$.
	Thus, either there exists a path going from $\beta_{k, j}$ to $\alpha_{k, j}$ and a path going from $\alpha_{k, j}$ to $\beta_{k, j'} = \beta_{k, j}$, and $\beta_{k, j}$ and $\alpha_{k, j}$ belong to the same orbit; or there exists a path going from $\alpha_{k, j}$ to $\beta_{k, j + 1}$ and a path going from $\beta_{k, j + 1}$ to $\alpha_{k, j'} = \alpha_{k, j}$, and $\alpha_{k, j}$ and $\beta_{k, j + 1}$ belong to the same orbit.
	
	Finally, let us focus on such an orbit $O$ with two out-gates $\alpha$ and $\beta$ such that $\mathrm{L}_{\alpha}(B) = \mathrm{L}_{\alpha_k}(A_k)$ and $\mathrm{L}_{\beta}(B) = \mathrm{L}_{\beta_k}(A_k)$.
	We know that for every $i \in \mathbb{N}$ such that $1 \leq i < k$, we have $\delta_k(\beta_k, b^i) = \beta_{k-i}$ with $|\mathrm{L}_{\beta_{k-i}}(A_k)| < \infty$.
	Since $\mathrm{L}_{\beta}(B) = \mathrm{L}_{\beta_k}(A_k)$ and $B$ is $(k-1)$-block deterministic, there exist $j \in \mathbb{N}$ and $p \in Q_B$ such that $1 \leq j < k$, $\delta_B(\beta, [b^j]) = p$ and $\mathrm{L}_p(B) = \mathrm{L}_{\beta_{k-j}}(A_k)$.
	This means that $|\mathrm{L}_p(B)| < \infty$, so $p \notin O$.
	Now, if there does not exist a state $q \in Q_B$ such that $\delta_B(\alpha, [b^j]) = q$, then $B$ does not have the orbit property.
	So, let us suppose that such a state exists.
	We know that for every $i \in \mathbb{N}$ such that $1 \leq i < k$, we have $\delta_k(\alpha_k, b^i) = \alpha_{k-i}$ with $|\mathrm{L}_{\alpha_{k-i}}(A_k)| = \infty$.
	Since $\mathrm{L}_{\alpha}(B) = \mathrm{L}_{\alpha_k}(A_k)$, we have $\mathrm{L}_q(B) = \mathrm{L}_{\alpha_{k-j}}(A_k)$ and $|\mathrm{L}_q(B)| = \infty$.
	So $p \neq q$ and $B$ does not have the orbit property.
	
	Since $\mathrm{L}(A_k)$ cannot be recognized by a $(k-1)$-block deterministic alphabetic image of an automaton passing the BKW test, following Theorem \ref{th:KBD} it holds that $\mathrm{L}(A_k)$ is not $(k-1)$-block deterministic.
\end{proof}

\section{A Witness for the Infinite Hierarchy of Lookahead Deterministic Languages}\label{se:kld_witness}

	In this section, we give a parameterized family $(\mathrm{L_j})_{j \geq 1}$ such that $\mathrm{L_j}$ is $(j+1)$-lookahead deterministic but not $j$-lookahead deterministic.
	In order to prove it, we show that any $j$-lookahead deterministic Glushkov automaton does not recognize $\mathrm{L_j}$.

	Let $j \in \mathbb{N}$ and let $A_j = (\Sigma, Q_j, I, F_j, \delta_j)$ be the automaton (given in Figure \ref{fg:Aj}) such that:
	\begin{itemize}
		\item $\Sigma = \{a\}$
		\item $Q_j = \{\alpha_i \mid 0 \leq i \leq 2j\}$
		\item $I = \{\alpha_0\}$
		\item $F_j = \{\alpha_0, \alpha_j\}$
		\item $\delta_j = \{(\alpha_i, a, \alpha_{i+1}) \mid 0 \leq i < 2j\} \cup \{(\alpha_{2j}, a, \alpha_0)\}$
	\end{itemize}

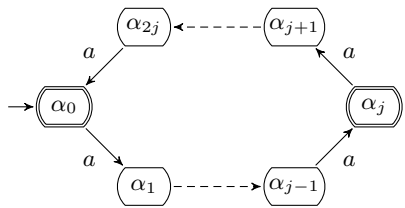
\begin{figure}[H]
	\centering
	
	\begin{tikzpicture}
		\node[state, initial, accepting] (a0) {$\alpha_0$};
		\node[state, below right of=a0] (a1) {$\alpha_1$};
		\node[state, right of=a1, xshift=0.5cm] (ajm1) {$\alpha_{j-1}$};
		\node[state, accepting, above right of=ajm1] (aj) {$\alpha_j$};
	    \node[state, above left of=aj] (ajp1) {$\alpha_{j+1}$};
		\node[state, left of=ajp1, xshift=-0.5cm] (a2j) {$\alpha_{2j}$};
	    \path[->]
		    (a0) edge [swap] node {$a$} (a1)
		    (a1) edge [densely dashed, swap] node {} (ajm1)
		    (ajm1) edge [swap] node {$a$} (aj)
		    (aj) edge [swap] node {$a$} (ajp1)
		    (ajp1) edge [densely dashed, swap] node {} (a2j)
	   	 	(a2j) edge [swap] node {$a$} (a0)
		;
	\end{tikzpicture}
	
	\caption{The minimal deterministic automaton $A_j$}
	\label{fg:Aj}
\end{figure}	

Let us first show that the languages in the family are distinct and satisfy the condition of lookahead determinism.

\begin{proposition}
	$\forall j \in \mathbb{N}$, $A_j$ is a minimal deterministic automaton.
\end{proposition}
\begin{proof}
	By construction, $\forall j \in \mathbb{N}$, $A_j$ is trimmed and deterministic.
	Then, if $j = 0$, there is only one state which is initial and final, so $A_0$ is minimal.
	Otherwise, if $j > 0$, $F_j = \{\alpha_0, \alpha_j\}$ such that $\{a^j\} \in \mathrm{L}_{\alpha_0}(A_j)$ and $\{a^j\} \notin \mathrm{L}_{\alpha_j}(A_j)$.
	Thus, $\alpha_0$ and $\alpha_j$ are not equivalent, and so are every non final states.
	Therefore, for every $j \in \mathbb{N}$, $A_j$ is also minimal.
\end{proof}

Thus, for all $j \neq j'$, since $|Q_j| \neq |Q_{j'}|$, then $\mathrm{L}(A_j) \neq \mathrm{L}(A_{j'})$.
Furthermore,

\begin{proposition}\label{prop:Aj_jp1_ld}
	$\forall j \in \mathbb{N}$, $\mathrm{L}(A_j)$ is $(j+1)$-lookahead deterministic.
\end{proposition}
\begin{proof}
	Let us consider the regular expression $E_j = (a^{2j+1})^* \cdot (\varepsilon + a^j)$.
	Then $E_j$ is $(j+1)$-lookahead deterministic
	And since the minimal deterministic automaton recognizing $\mathrm{L}(E_j)$ is isomorphic to $A_j$, then $\mathrm{L}(A_j) = \mathrm{L}(E_j)$.
	So $\mathrm{L}(A_j)$ is $(j+1)$-lookahead.
\end{proof}

	Let $j' \in \mathbb{N} \setminus \{0\}$ and let $G = (\{a\}, Q_G, {i_G}, F_G, \delta_G)$ be a $j'$-lookahead deterministic Glushkov automaton.
	We demonstrate that $G$ cannot recognized $\mathrm{L}(A_{j'})$, that is to say that $\mathrm{L}(A_{j'})$ is not $j'$-lookahead deterministic.
 
	In order to do so, we consider a property of Glushkov automata from Proposition 4.2 of \cite{CZ97}:
\begin{lemma}\label{lm:CZ_stability}
	Let $O \in \mathcal{O}_G$ be a non-trivial orbit of $G$.
	Then for every in-gate $o_i$ of $O$ and for every out-gate $o_o$ of $O$, $o_i \in \delta(o_o, a)$.
\end{lemma}

	Let us first restrain the set of Glushkov automata to consider.
	We show that a state in a lookahead deterministic unary automaton cannot admit two distinct successors with infinite right languages.

\begin{proposition}\label{prop:klau_Lq_infinite}
	$\forall s \in Q_G, \forall q_1, q_2 \in \delta_G(s, a), (|\mathrm{L}_{q_1}(G)| = \infty \wedge 	|\mathrm{L}_{q_2}(G)| = \infty) \Longrightarrow q_1 = q_2$.
\end{proposition}
\begin{proof}
	Let us suppose that there exist 3 states $q, q_1, q_2 \in Q_G$ such that $q_1 \neq q_2$, $|\mathrm{L}_{q_1}(G)| = \infty$, $|\mathrm{L}_{q_2}(G)| = \infty$ and ${q_1, q_2} \subset \delta_G(q, a)$.
	Then, necessarily, $\delta_G(q_1, a^{j'-1}) \neq \emptyset$ and $\delta_G(q_2, a^{j'-1}) \neq \emptyset$.
	Thus $G$ is not $j'$-lookahead deterministic.
\end{proof}

\begin{corollary}\label{prop:klau_orbite_non_triviale}
	Among the orbit of $G$, at most one is non-trivial.
\end{corollary}

	Furthermore, let us suppose that $G$ has no non-trivial orbit, then $|\mathrm{L}(G)| < \infty$.
But, since for every $j' \in \mathbb{N} \setminus \{0\}$, we have $|\mathrm{L}(A_{j'})| = \infty$, then $G$ could not recognize $\mathrm{L}(A_{j'})$.
Thus, $G$ must have a single non-trivial orbit denoted by $O$, of size $l_O = |O|$.
Moreover, the gates of $O$ are remarkable:

\begin{proposition}\label{prop:klau_portes_uniques}
	$O$ has a single in-gate and a single out-gate.
\end{proposition}
\begin{proof}
	The fact that there exists a single in-gate is also a direct consequence of Proposition \ref{prop:klau_Lq_infinite}.
	Furthermore, let us suppose that there exist 2 distinct states $g_1, g_2 \in \mathrm{G_{out}}(O)$.
	Let $o_{in} \in \mathrm{G_{in}}(O)$ be the single in-gate of $O$.
	Since $G$ is a Glushkov automaton, following Lemma~\ref{lm:CZ_stability}, $o_{in} \in \delta_G(g_1, a) \cap \delta_G(g_2, a)$.
	Consequently, there exists $s \in O$ such that $|\delta_G(s, a)| > 1$, contradicting Proposition~\ref{prop:klau_Lq_infinite}.
\end{proof}

	Consequently, we denote by $o_{in}$ the single in-gate of $O$ and by $o_{out}$ its single out-gate.

	As a corollary of Proposition \ref{prop:klau_Lq_infinite}, and since Glushkov automata are standard, then there exists a single state $s \in Q_G \setminus O$ such that $o_{in} \in \delta_G(s, a)$, which can be reached from the initial state by a single word $w_s$ such that $|w_s| = m$.
	This allows us to characterize the words reaching $o_{out}$ from $i$.

\begin{lemma}\label{lm:ooPathLength}
	$\forall w \in \{a\}^*, o_{out} \in \delta_G(i, w) \Longleftrightarrow \exists k \in \mathbb{N} \setminus \{0\}, |w| = m + k \times l_O$.
\end{lemma}
\begin{proof}
	Following Proposition \ref{prop:klau_Lq_infinite}, we can deduce: $\forall o \in O, |\delta_G(o, a) \cap O| = 1$.
	Thus, $\forall o \in O, o \in \delta_G(o, w) \Longleftrightarrow \exists k \in \mathbb{N}, |w| = k \times l_O$.
	And following Lemma \ref{lm:CZ_stability}, $o_{in} \in \delta_G(o_{out}, a)$ and thus $o_{out} \in \delta_G(o_{in}, w) \Longleftrightarrow \exists k \in \mathbb{N}, |w| = l_O - 1 + k \times l_O$.
	Since $o_{in} \in \delta_G(s, a)$ and $s \in \delta_G(i, w) \Longleftrightarrow |w| = m$, then $o_{out} \in \delta_G(i, w) \Longleftrightarrow \exists k \in \mathbb{N}, |w| = m + 1 + l_O - 1 + k \times l_O$.
\end{proof}

	Moreover, we give a necessary condition over words in $\mathrm{L}(A_{j'})$.

\begin{lemma}\label{lm:w123}
	Let $w_1, w_2, w_3 \in \mathrm{L}(A_{j'})$ such that $|w_1| < |w_2| < |w_3|$ and for every $w \in \mathrm{L}(A_{j'})$ such that $w \neq w_2$, either $|w| \leq |w_1|$ or $|w| \geq |w_3|$.
	Then, either $|w_3| - |w_2| = j'+1$ and $|w_2| - |w_1| = j'$, or $|w_3| - |w_2| = j'$ and $|w_2| - |w_1| = j' + 1$.
\end{lemma}
\begin{proof}
	Since $\mathrm{L}(A_{j'}) = \mathrm{L}((a^{2j'+1})^* \cdot (\varepsilon + a^{j'})) = \{\varepsilon, a^{j'}, a^{2j'+1}, a^{3j'+1}, a^{4j'+2}, \cdots \}$, then either $|w_3| - |w_2| = j'+1$ and $|w_2| - |w_1| = j'$, or $|w_3| - |w_2| = j'$ and $|w_2| - |w_1| = j' + 1$.
\end{proof}
	
	From the two previous lemmas, we show that $G$ cannot recognize $\mathrm{L}(A_{j'})$.
	
\begin{proposition}
	$\mathrm{L}(G) \neq \mathrm{L}(A_{j'})$.
\end{proposition}
\begin{proof}
	Let $Q_o = \delta(o_{out}, a) \setminus O$ be the set of direct successors of $o_{out}$ outside of $O$, and $\mathrm{L_o} = \bigcup_{q \in Q_o} \mathrm{L}_q(A)$ be the union of their right languages.
	Then, since $G$ is $j'$-lookahead deterministic, the length of any word of $\mathrm{L_o}$ is strictly smaller than $j' - 1$.
	
	Let us consider the set $\mathrm{L_{out}}$ of words reaching a final state from $o_{out}$ without going through $O$.
	By definition, $\mathrm{L_{out}} = \{a\} \cdot \mathrm{L_o} \cup \{\varepsilon\}$ if $o_{out} \in F$, $\{a\} \cdot \mathrm{L_o}$ otherwise.
	Then, the length of any word $w \in \mathrm{L_{out}}$ is strictly smaller than $j'$.
	
	If $\mathrm{L_{out}} = \emptyset$, then $|\mathrm{L}(G)| < \infty$ and $\mathrm{L}(G) \neq \mathrm{L}(A_{j'})$.
	
	Now, let us suppose that there exist two distinct words $w_{o1}, w_{o2} \in \mathrm{L_{out}}$ such that $|w_{o2}| > |w_{o1}|$, then $|w_{o2}| - |w_{o1}| < j'$.
	Moreover, there exists a word $w \in \{a\}^*$ such that $o_{out} \in \delta(i, w)$.
	Thus, $\{ww_{o1}, ww_{o2}\} \subset \mathrm{L}(G)$ and $|ww_{o1} - ww_{o2}| < j'$.
	Following Lemma \ref{lm:w123}, $\mathrm{L}(G) \neq \mathrm{L}(A_{j'})$.
	
	Now, let us suppose that $\mathrm{L_{out}} = \{w_o\}$.
	Since $O$ is the single non-trivial orbit and $o_{out}$ is its single out-gate, then there exists $n \in \mathbb{N}$ such that for every word $w_G \in \mathrm{L}(G)$, if $|w_G| \geq n$ then $w_G = w_pw_o$ with $o_{out} \in \delta_G(i, w_p)$.
	Let $w_{G1}, w_{G2}, w_{G3} \in \mathrm{L}(G)$ such that $n \leq |w_{G1}| < |w_{G2}| < |w_{G3}|$ and for every $w_G \in\mathrm{L}(G)$ such that $w_G \neq w_{G2}$, either $|w_G| \leq |w_{G1}|$ or $|w_G| \geq |w_{G3}|$.
	Then, following Lemma \ref{lm:ooPathLength}, $|w_{G3}| - |w_{G2}| = |w_{G2}| - |w_{G1}| = l_O$, which means that $\mathrm{L}(G) \neq \mathrm{L}(A_{j'})$.
\end{proof}

	Thus, $\mathrm{L}(A_{j'})$ cannot be recognized by a $j'$-lookahead deterministic Glushkov automaton.
	Consequently:
	
\begin{proposition}\label{prop:Aj_non_j_ld}
	$\forall j \in \mathbb{N} \setminus \{0\}$, $\mathrm{L}(A_j)$ is not $j$-lookahead deterministic.
\end{proposition}

	We can conclude that:
\begin{theorem}
	There is a proper hierarchy in $k$-lookahead deterministic regular languages.
\end{theorem}

\section{Relationship Between Block Deterministic Languages and Lookahead Deterministic Languages}\label{se:inclusion_kbd_kld}

	Han and Wood stated in~\cite{HW08} that block deterministic languages are a proper subfamily of lookahead deterministic languages.
	However, in their proof of the block deterministic languages being a proper subfamily of the lookahead deterministic languages, they used a statement made by Giammarresi \emph{et al.} in~\cite{GMW01} about the family of languages $\mathrm{L}((a+b)^*a(a+b)^{k-1})$ not being $k$-block deterministic.
	But we proved that Lemma \ref{le:GMWH}, which is used as a basis for deciding if a language is block deterministic, is wrong.
	So we cannot be sure that their example is not block deterministic and give our own proof.

\subsection{Block Deterministic Languages as a Subfamily of Lookahead Deterministic Languages}

	We start by presenting some properties of block regular expressions and regular expressions over the language of their marked expressions.

\begin{lemma}
	A block regular expression $E_b$ is $k$-block deterministic if and only if: 
	$\forall u,v,w \in \Pi_{E_b}^*, \forall b_1,b_2 \in \Pi_{E_b}, (ub_1v, ub_2w \in \mathrm{L}(E_b^{\sharp}) \wedge b_1 \neq b_2) \Longrightarrow b_1^{\natural} \notin \mathrm{Pref}(b_2^{\natural})$.
\end{lemma}
\begin{proof}
	Giammarresi \emph{et al.} in \cite{GMW01} defined a block regular expression $E_b$ being block deterministic if the following two conditions hold:
	\begin{itemize}
		\item $\forall b_1, b_2 \in \mathrm{First}(E_b^{\sharp}), b_1 \neq b_2 \Longrightarrow b_1^{\natural} \notin \mathrm{Pref}(b_2^{\natural})$
		\item $\forall x \in \Pi_{E_b}, \forall b_1, b_2 \in \mathrm{Follow}(E_b^{\sharp}, x), b_1 \neq b_2 \Longrightarrow b_1^{\natural} \notin \mathrm{Pref}(b_2^{\natural})$.
	\end{itemize}
	And we can deduce from it the aforementioned property.
\end{proof}

\begin{lemma}[\cite{HW08}]
	A regular expression $E$ is $k$-lookahead deterministic if and only if:
	$\forall u,v,w \in \Pi_E^*, \forall x,y \in \Pi_E^k, (uxv, uyw \in \mathrm{L}(E^{\sharp}) \wedge x(1) \neq y(1) \Longrightarrow x^{\natural} \neq y^{\natural}$
\end{lemma}

	Let $E_b$ be a ($\Sigma, k$)-block regular expression over an alphabet $\Gamma$ and $E_b^{\sharp}$ its marked block regular expression  over the alphabet $\Pi_{E_b}$.
	Let $\Omega = \{a_{i, j} \mid \exists [w]_i \in \Pi_{E_b}, w[j] = a\}$ be an alphabet, $\varphi : \Pi_{E_b} \rightarrow B_{\Omega, k}$ a function such that for every $[w]_i \in \Pi_{E_b}$, $\varphi([w]_i) = w[1]_{i, 1} \cdot w[2]_{i, 2} \cdots w[|w|]_{i, |w|}$.
	
\begin{example}
	$\varphi([aba]_5) = a_{5, 1}b_{5, 2}a_{5, 3}$.
\end{example}
	
	Every symbol of $\Omega$ is linked to only one position of block of $\Pi_{E_b}$ and represents the position of a letter in a position of block.
	Considering that every block of $\Pi_{E_b}$ is indexed differently, every element of $\Omega$ is produced by only one block.
	We define the following functions for every $a_{i, j} \in \Omega$ :
	\begin{itemize}
		\item $\mathrm{Pos}_{E_b} : \Omega \rightarrow \Pi_{E_b}$ with $\mathrm{Pos}_{E_b}(a_{i, j}) =  [w]_i$
		\item $\mathrm{BlockPos}_{E_b} : \Omega \rightarrow \mathbb{N} \setminus \{0\}$ with $\mathrm{BlockPos}_{E_b}(a_{i, j}) = j$
		\item $\mathrm{BlockLength}_{E_b} : \Omega \rightarrow \mathbb{N} \setminus \{0\}$ with $\mathrm{BlockLength}_{E_b}(a_{i, j}) = |(\mathrm{Pos}_{E_b}(a_{i, j}))^{\natural}|$
	\end{itemize}
	
	A word $w \in \Omega^*$ is \emph{simple block complete} if there exists $x \in \Pi_{E_b}$ such that $\varphi(x) = w$.
	The set of simple block complete words over $\Omega$ is denoted by $SBC(\Omega)$.
	Let $w \in (SBC(\Omega))^*$, then $w$ is \emph{block complete} and the set of block complete words is denoted by $BC(\Omega)$.
	Then, we can define $\varphi$ as a bijection between $\Pi_{E_b}$ and $SBC(\Omega)$, and extend it by morphism between $\Pi_{E_b}^*$ and $BC(\Omega)$.
	
	Let $\chi$ be the function which takes a marked block regular expression $E_b^{\sharp}$ and transform it into a marked regular expression $E^{\sharp}$ over the alphabet $\Omega$ such that, for every subexpression $F = x \in \Pi_{E_b}$, $\chi(F) = \varphi(x)[1] \cdot \varphi(x)[2] \cdots \varphi(x)[|\varphi(x)|]$.
	In this way, we extend the definition of marked regular expresion to regular expressions over an alphabet indexed by any number of element, and whose symbol appear only once in the regular expression.
	
	\begin{example}
		Let $E_b = ([aba] + [abb])^*[aa]$, then $E_b^{\sharp} = ([aba]_1 + [abb]_2)^*[aa]_3$ and $E^{\sharp} = \chi(E_b^{\sharp}) = (a_{1, 1}b_{1, 2}a_{1, 3} + a_{2, 1}b_{2, 2}b_{2, 3})^*a_{3, 1}a_{3, 2}$.
	\end{example}
	
	In the rest of this section, we consider the following regular expression : $E_b$ a ($\Sigma, k$)-block regular expresion, its marked block regular expression $E_b^{\sharp}$ over $\Pi_{E_b}$, $E^{\sharp} = \chi(E_b^{\sharp})$ a marked regular expression over $\Omega$ and $E = (E^{\sharp})^{\natural}$ a regular expression over $\Sigma$.

	We can deduce some obvious properties about $E^{\sharp}$:
	\begin{gather}
		\mathrm{Null}(E^{\sharp}) = \mathrm{Null}(E_b^{\sharp}) \label{claim:null}\\
		\mathrm{First}(E^{\sharp}) = \{a_{i, 1} \in \Omega \mid \mathrm{Pos}_{E_b}(a_{i, 1}) \in \mathrm{First}(E_b^{\sharp}\} \label{claim:first}\\
		\mathrm{Last}(E^{\sharp}) = \{a_{i, j} \in \Omega \mid \mathrm{Pos}_{E_b}(a_{i, j}) \in \mathrm{Last}(E_b^{\sharp}) \wedge j = \mathrm{BlockLength}_{E_b}(a_{i, j})\} \label{claim:last}\\
		\mathrm{Follow}(E^{\sharp}, a_{i, j}) = 
			\begin{cases}
				\{b_{i, j+1} \in \Omega\}\\
				\quad \text{if}\ j < \mathrm{BlockLength}_{E_b}(a_{i, j})\\
				\{b_{i^{\prime}, 1} \in \Omega \mid \mathrm{Pos}_{E_b}(b_{i^{\prime}, 1}) \in \mathrm{Follow}(E_b^{\sharp}, \mathrm{Pos}_{E_b}(a_{i, j}))\}\\
				\quad \text{otherwise}
			\end{cases}	\label{claim:follow}
	\end{gather}

	We can also deduce some properties about block complete words, such that for every $w \in \Omega^*$, $w$ is bloc complete if $w = \varepsilon$ or if it validates all of the following conditions:
	\begin{gather}
		\mathrm{BlockPos}_{E_b}(w[1]) = 1  \label{claim:CB_BP1}\\
		\mathrm{BlockPos}_{E_b}(w[|w|]) = \mathrm{BlockLength}_{E_b}(w[|w|])  \label{claim:CB_BPBL}
	\end{gather}
	\begin{multline}
		\forall i \in [1, |w|[, \mathrm{BlockPos}_{E_b}(w[i+1]) = (\mathrm{BlockPos}_{E_b}(w[i]) + 1)\\
		\Longrightarrow \mathrm{Pos}_{E_b}(w[i+1]) = \mathrm{Pos}_{E_b}(w[i])  \label{claim:CB_BP_eq}
	\end{multline}
	\begin{multline}
		\forall i \in [1, |w|[, \mathrm{BlockPos}_{E_b}(w[i+1]) \neq (\mathrm{BlockPos}_{E_b}(w[i]) + 1) \Longrightarrow \\
		\mathrm{BlockPos}_{E_b}(w[i]) = \mathrm{BlockLength}_{E_b}(w[i]) \wedge \mathrm{BlockPos}_{E_b}(w[i+1]) = 1 \label{claim:CB_BP_neq}
	\end{multline}

	Let us show that the previous statements induce important conditions concerning $\mathrm{L}(E^{\sharp})$.

\begin{proposition}\label{prop:mot_bc}
	Every word $w \in \mathrm{L}(E^{\sharp})$ is block complete.
\end{proposition}
\begin{proof}
	Let $w \in \mathrm{L}(E^{\sharp})$.
	If $w = \varepsilon$, then $w$ is block complete.
	Now let us suppose that $w \neq \varepsilon$.
	Following property \eqref{claim:first}, $w$ validates property \eqref{claim:CB_BP1}.
	Following property \eqref{claim:last}, $w$ validates property \eqref{claim:CB_BPBL}.
	And following property \eqref{claim:follow}, $w$ validates the properties \eqref{claim:CB_BP_eq} and \eqref{claim:CB_BP_neq}.
	Therefore $w$ is block complete.
\end{proof}

\begin{proposition}\label{prop:fourche}
	Let $u, v, w, \in \Omega^*$ and $x, y \in \Omega$ such that $uxv, uyw \in \mathrm{L}(E^{\sharp})$ and $x \neq y$.
	Then $u, xv, yw \in \mathrm{Subw}(\mathrm{L}(E^{\sharp})) \cap BC(\Omega)$ and $\mathrm{Pos}_{E_b}(x) \neq \mathrm{Pos}_{E_b}(y)$.
\end{proposition}
\begin{proof}
	Since the words $uxv, uyw \in \mathrm{L}(E^{\sharp})$, then $u, xv, yw \in \mathrm{Subw}(\mathrm{L}(E^{\sharp}))$.
	Following property \eqref{claim:follow}, $|\mathrm{Follow}(E^{\sharp}, u[|u|])| > 1$ implies that $\mathrm{BlockPos}_{E_b}(u[|u|]) =$ $\mathrm{BlockLength}_{E_b}(u[|u|])$ and $\mathrm{BlockPos}_{E_b}(x) = \mathrm{BlockPos}_{E_b}(y) = 1$.
	Then, using Proposition \ref{prop:mot_bc}, we can conclude that $u, xv, yw \in BC(\Omega)$.
	Moreover, since $\mathrm{BlockPos}_{E_b}(x) = \mathrm{BlockPos}_{E_b}(y) = 1$ and $x \neq y$, then $\mathrm{Pos}_{E_b}(x) \neq \mathrm{Pos}_{E_b}(y)$.
\end{proof}

\begin{proposition}\label{prop:bijectionLangage}
	$\varphi$ is a bijection between $\mathrm{L}(E_b^{\sharp})$ and $\mathrm{L}(E^{\sharp})$.
\end{proposition}
\begin{proof}
	Following property \eqref{claim:null}, $\varepsilon \in \mathrm{L}(E^{\sharp})\Longleftrightarrow \varepsilon \in \mathrm{L}(E_b^{\sharp})$.
	Since $\varphi(\varepsilon) = \varepsilon$, then $\varepsilon \in \mathrm{L}(E^{\sharp})\Longleftrightarrow \varphi(\varepsilon) \in \mathrm{L}(E_b^{\sharp})$.

	Let $w_B \in \mathrm{L}(E_b^{\sharp})$ such that $w_B \neq \varepsilon$, according to the definition of function $\chi$, $\varphi(w_B) \in \mathrm{L}(E^{\sharp})$.
	
	Let $w \in \mathrm{L}(E^{\sharp})$ such that $w \neq \varepsilon$.
	Following Proposition \ref{prop:mot_bc}, $w$ is block complete.
	Let set $w = w_1 \cdot w_2 \cdots w_n$ such that for every $i \in [1, n]$, $w_i \in SBC(\Omega)$.
	Then:
	\begin{itemize}
		\item $w_1[1] \in \mathrm{First}(E^{\sharp})$
		\item $w_n[|w_n|] \in \mathrm{Last}(E^{\sharp})$
		\item $\forall i \in [1, n[, w_{i+1}[1] \in \mathrm{Follow}(E^{\sharp}, w_i[|w_i|])$.
	\end{itemize}
	This means that :
	\begin{itemize}
		\item $\varphi^{-1}(w_1) \in \mathrm{First}(E_b^{\sharp})$ (following property \eqref{claim:first})
		\item $\varphi^{-1}(w_n) \in \mathrm{Last}(E_b^{\sharp})$ (following property \eqref{claim:last})
		\item $\forall j \in [1, n[, \varphi^{-1}(w_{i+1}) \in \mathrm{Follow}(E_b^{\sharp}, \varphi^{-1}(w_i))$ (following property \eqref{claim:follow}).
	\end{itemize}
	Therefore, $\varphi^{-1}(w) =  \varphi^{-1}(w_1) \cdot \varphi^{-1}(w_2) \cdots \varphi^{-1}(w_n) \in \mathrm{L}(E_b^{\sharp})$.
	This demonstrates that $\varphi$ is also a bijection between $\mathrm{L}(E_b^{\sharp}) \rightarrow \mathrm{L}(E^{\sharp})$.
\end{proof}

In the same way, we can also prove that $\varphi$ is a bijection between $\mathrm{Subw}(\mathrm{L}(E_b^{\sharp}))$ and $\mathrm{Subw}(\mathrm{L}(E^{\sharp})) \cap BC(\Omega)$.

\begin{proposition}\label{prop:preservationLangage}
	$\mathrm{L}(E_b) = \mathrm{L}(E)$.
\end{proposition}
\begin{proof}
	Let $u \in \Gamma^*, w \in \Omega^*, w_B \in \Pi_{E_b}^*$ such that $u = w_B^{\natural}$ and $w = \varphi(w_B)$.
	According to the definition of the bijection $\varphi$, for every $w_B \in \Pi_{E_b}, w_B^{\natural} = (\varphi(w_B))^{\natural}$.
	And since $\Gamma^* \subset \Sigma^*$, then:\\
	$\begin{array}{ccl}
   		u \in \mathrm{L}(E_b) & \iff & w_B \in \mathrm{L}(E_b^{\sharp})\\
  		& \iff & \varphi(w_B) = w \in \mathrm{L}(E^{\sharp})\\
	  	& \iff & w^{\natural} = \varphi(w_B))^{\natural} = w_B^{\natural} = u \in \mathrm{L}(E)\\
	\end{array}$

\end{proof}

\begin{theorem}\label{th:kbd_kla}
	If $E_b$ is $k$-block deterministic, then $E$ is $k$-lookahead deterministic.
\end{theorem}
\begin{proof}
	Let us suppose that $A$ is not $k$-lookahead deterministic.
	Then, there exist $u, v, w, \in \Omega^*$ and $l_1, l_2 \in \Omega^k$ such that $ul_1v, ul_2w \in \mathrm{L}(E^{\sharp})$, $l_1[1] \neq l_2[1]$ et $(l_1)^{\natural} = (l_2)^{\natural}$.
	Since $ul_1v, ul_2w \in \mathrm{L}(E^{\sharp})$ then, following Proposition \ref{prop:fourche}, $u, l_1v, l_2w \in \mathrm{Subw}(\mathrm{L}(E^{\sharp})) \cap BC(\Omega)$ and $\mathrm{Pos}_{E_b}(l_1[1]) \neq \mathrm{Pos}_{E_b}(l_2[1])$.
	Let $u_B, v_B, w_B \in \Pi_{E_b}^*$ and $b_1, b_2 \in \Pi_{E_b}$ such that $u_B = \varphi^{-1}(u), b_1v_B = \varphi^{-1}(l_1v), b_2w_B = \varphi^{-1}(l_2w)$.
	Then $u_Bb_1v_B, u_Bb_2w_B \in \mathrm{L}(E_b^{\sharp})$.
	Since $\mathrm{Pos}_{E_b}(l_1[1]) \neq \mathrm{Pos}_{E_b}(l_2[1])$, then $b_1 \neq b_2$.
	And since $b_1v_B = \varphi^{-1}(l_1v), b_2w_B = \varphi^{-1}(l_2w)$, then $(b_1v_B)^{\natural} = (l_1v)^{\natural}, (b_2w_B)^{\natural} = (l_2w)^{\natural}$.
	But since $E_b$ is $k$-block deterministic, then $|b_1^{\natural}| \leq k$ and $|b_2^{\natural}| \leq k$.
	So $b_1^{\natural} \in \mathrm{Pref}(l_1^{\natural})$ and $b_2^{\natural} \in \mathrm{Pref}(l_2^{\natural})$.
	But since $l_1^{\natural} = l_2^{\natural}$, then either $b_1^{\natural} \in \mathrm{Pref}(b_2^{\natural})$, or $b_2^{\natural} \in \mathrm{Pref}(b_1^{\natural})$.
	Thus $E_b$ cannot be $k$-block deterministic.
	Therefore, if $E_b$ is $k$-block deterministic, then $E$ is indeed $k$-lookahead deterministic.
\end{proof}

To conclude, following Theorem \ref{th:kbd_kla} and Proposition \ref{prop:preservationLangage} :

\begin{theorem}
	The family of $k$-block deterministic languages is included in the family of $k$-lookahead deterministic languages.
\end{theorem}

\subsection{Block Deterministic Languages as a Proper Subfamily of Lookahead Deterministic Languages}

	In this section, we show that the family block deterministic languages is strictly included in the one of lookahead deterministic languages.
	We first show that unary $k$-block deterministic languages are one-unambiguous, and then conclude using the parameterized family provided in Section \ref{se:kld_witness}.

\begin{proposition}
	Let $E$ be a regular expression over an unary alphabet $\{a\}$.
	If $E$ is $k$-block deterministic, then $|\mathrm{First}(E^{\sharp})| \leq 1$ and $\forall x \in \Pi_E, |\mathrm{Follow}(E^{\sharp}, x)| \leq 1$.
\end{proposition}
\begin{proof}
	Let us suppose that $E$ is $k$-block deterministic.
	If $|\mathrm{First}(E^{\sharp})| > 1$, then there exist $b_1, b_2 \in \mathrm{First}(E^{\sharp})$ such that $b_1 \neq b_2$.
	And since $E$ is defined over an unary alphabet, then necessarily, either $(b_1)^{\natural} \in \mathrm{Pref}((b_2)^{\natural})$, or $(b_2)^{\natural} \in \mathrm{Pref}((b_1)^{\natural})$, which means that $E$ is not $k$-block deterministic.
	Therefore $|\mathrm{First}(E^{\sharp})| \leq 1$, and the same reasoning can be applied to $\mathrm{Follow}(E^{\sharp}, x)$.
\end{proof}

\begin{lemma}
	Let $\Sigma$ be an alphabet and $E_b$ be a $k$-block deterministic regular expression over an alphabet $\Gamma \subset B_{\Sigma, k}$.
	If $\forall x \in \Pi_{E_b}, | \mathrm{Follow}(E_b^{\sharp}, x)| \leq 1$ and  $|\mathrm{First}(E_b^{\sharp})| \leq 1$ , then $\mathrm{L}(E_b)$ is $1$-block deterministic.
\end{lemma}
\begin{proof}
	Let $E^{\sharp} = \chi(E_b^{\sharp})$ be a marked regular expression and $E = (E^{\sharp})^{\natural}$ a regular expression over $\Sigma$.
	Following Proposition \ref{prop:preservationLangage}, $\mathrm{L}(E_b) = \mathrm{L}(E)$.
	Now, following property \eqref{claim:first}, $|\mathrm{First}(E_b^{\sharp})| \leq 1$ imply that $|\mathrm{First}(E^{\sharp})| \leq 1$.
	And following property \eqref{claim:follow}, $\forall x \in \Pi_{E_b}, |\mathrm{Follow}(E_b^{\sharp}, x)| \leq 1$ imply that $\forall y \in \Pi_E, |\mathrm{Follow}(E^{\sharp}, y)| \leq 1$.
	Therefore, following the construction of Glushkov automata in Defintion \ref{def:Glushkov}, $G_E$ is deterministic and $E$ is one-unambiguous (that is to say $1$-block deterministic).
\end{proof}

	Consequently:

\begin{theorem}
	If an unary language is block deterministic, then it is $1$-block deterministic.
\end{theorem}

	Therefore:

\begin{proposition}
	For any $k \in \mathbb{N} \setminus \{0, 1\}$, there exists unary languages which are $k$-lookahead deterministic without being block deterministic.	
\end{proposition}
\begin{proof}
	In Section \ref{se:kld_witness}, we showed that for any $j \in \mathbb{N} \setminus \{0\}$, $\mathrm{L}(A_j)$ is $(j+1)$-lookahead deterministic without being $j$-lookahead deterministic.
	And since they are not $1$-lookahead deterministic (that is to say one-unambiguous), they are not block deterministic for any $k$.
	Thus, for any $k \in \mathbb{N} \setminus \{0\}$, $\mathrm{L}(A_k)$ is $(k+1)$-lookahead deterministic without being block deterministic.
\end{proof}

	Finally:

\begin{theorem}
	$\forall k \in \mathbb{N} \setminus \{0\}$, the family of $k$-block deterministic languages is strictly included in the one of $k$-lookahead deterministic languages.
\end{theorem}

\section{Conclusion and Perspectives}

In this paper, we demonstrated that despite some erroneous results, there exists an infinite hierarchy in block deterministic languages.
We showed that such an infinite hierarchy also exists in lookahead deterministic languages.
And finally, showing that block-deterministic unary languages are one-unambiguous, we gave our own proof of the family of block deterministic languages being strictly included in the family of lookahead deterministic languages.

From these results, one can wonder whether there exists a $k$-lookahead deterministic language which is $(k+1)$-block deterministic without being $k$-block deterministic.
Another open problem is the decidability of the lookahead determinism of a language.
Finally, the decidability of the block determinism of a language has been studied by Giammarresi \emph{et al.} but proved with Lemma~\ref{le:GMWH} which we invalidated.
Thus, this problem is still open.

\bibliography{biblioRapport}

\begin{thebibliography}{10}

\bibitem{BW98}
Anne Br{\"{u}}ggemann{-}Klein and Derick Wood.
\newblock One-unambiguous regular languages.
\newblock {\em Inf. Comput.}, 140(2):229--253, 1998.

\bibitem{CZ97}
P.~Caron and D.~Ziadi.
\newblock Characterization of {G}lushkov automata.
\newblock {\em Theoret. Comput. Sci.}, 233(1--2):75--90, 2000.

\bibitem{CFM14}
Pascal Caron, Marianne Flouret, and Ludovic Mignot.
\newblock (k, l)-unambiguity and quasi-deterministic structures: An alternative
  for the determinization.
\newblock In {\em LATA}, pages 260--272, 2014.

\bibitem{GMW01}
Dora Giammarresi, Rosa Montalbano, and Derick Wood.
\newblock Block-deterministic regular languages.
\newblock In {\em Theoretical Computer Science, 7th Italian Conference, {ICTCS}
  2001, Torino, Italy, October 4-6, 2001, Proceedings}, pages 184--196, 2001.

\bibitem{Glu61}
V.~M. Glushkov.
\newblock The abstract theory of automata.
\newblock {\em Russian Mathematical Surveys}, 16:1--53, 1961.

\bibitem{HW08}
Yo-Sub Han and Derick Wood.
\newblock Generalizations of 1-deterministic regular languages.
\newblock {\em Inf. Comput.}, 206(9-10):1117--1125, 2008.

\bibitem{Hop71}
J.~E. Hopcroft.
\newblock An $n$ log $n$ algorithm for minimizing the states in a finite
  automaton.
\newblock In Z.~Kohavi, editor, {\em The theory of machines and computations},
  pages 189--196. Academic Press, New York, 1971.

\bibitem{Kle56}
S.~Kleene.
\newblock Representation of events in nerve nets and finite automata.
\newblock {\em Automata Studies}, Ann. Math. Studies 34:3--41, 1956.
\newblock Princeton U. Press.

\bibitem{Moo56}
E.~F. Moore.
\newblock Gedanken experiments on sequential machines.
\newblock In {\em Automata studies}, pages 129--153. Princeton Univ. Press,
  Princeton, N.J., 1956.

\bibitem{RS59}
M.~O. Rabin and D.~Scott.
\newblock Finite automata and their decision problems.
\newblock {\em IBM J. Res.}, 3(2):115--125, 1959.

\end{thebibliography}

\end{document}